\documentclass[11pt]{article}
\usepackage[utf8]{inputenc}
\usepackage[margin=1in]{geometry}
\usepackage{amsmath}
\usepackage{amssymb}
\usepackage{amsfonts}
\usepackage{amsthm}
\usepackage{bm}
\usepackage{color}
\usepackage{enumitem}
\usepackage[dvipsnames]{xcolor}

\usepackage[colorlinks,linkcolor=blue,citecolor=blue]{hyperref}
\usepackage[square,numbers,sort&compress]{natbib}

\usepackage{comment} 
\usepackage{subcaption}
\usepackage{pgf,tikz,pgfplots}
\usepackage{mathrsfs}
\newcommand{\s}{\mathbf{s}}

\newcommand{\E}{\mathop{\mathbb{E}}}
\newcommand{\opt}{\mathsf{OPT}}
\renewcommand{\v}{\mathbf{v}}
\newcommand{\kgnote}[1]{{\color{orange}{[Kira: #1]}}}
\newcommand{\alon}[1]{{\color{Green}{[Alon: #1]}}}
\newcommand{\shuran}[1]{{\color{RoyalBlue}{[Shuran: #1]}}}

\newtheorem{definition}{Definition}[section]
\newtheorem{lemma}{Lemma}[section]
\newtheorem{theorem}{Theorem}[section]

\newtheorem{example}{Example}[section]
\newtheorem{proposition}{Proposition}[section]

\title{Private Interdependent Valuations}

\author{
	Alon Eden\thanks{Harvard University; {\tt aloneden@seas.harvard.edu}. Supported by NSF Award IIS-2007887.}
	\and
	Kira Goldner \thanks{Boston University; {\tt goldner@bu.edu}. Supported by a Shibulal Family Career Development Professorship.}
	\and
	Shuran Zheng \thanks{Harvard University; {\tt shuran\_zheng@seas.harvard.edu}. Supported by NSF Award IIS-2007887.}
}
\date{}

\begin{document}

\maketitle
\begin{abstract}
    We consider the single-item \emph{interdependent value setting}, where there is a single item sold by a monopolist, $n$ buyers, and each buyer has a private signal $s_i$ describing a piece of information about the item.  Additionally, each bidder $i$ has a valuation function $v_i(s_1,\ldots,s_n)$ mapping the (private) signals of all buyers into a positive real number representing their value for the item.  This setting captures scenarios where the item's information is asymmetric or dispersed among agents, such as in competitions for oil drilling rights, or in auctions for art pieces.  Due to the increased complexity of this model compared to the standard private values model, it is generally assumed that each bidder's valuation function $v_i$ is \textit{public knowledge} to the seller or all other buyers. But in many situations, the seller may not know the bidders' valuation functions---how a bidder aggregates signals into a valuation is often their private information. In this paper, we design mechanisms that guarantee approximately-optimal social welfare while satisfying ex-post incentive compatibility and individually rationality for the case where the valuation functions are private to the bidders, and thus may be strategically misreported to the seller. 
	
	When the valuations are public, it is possible for optimal social welfare to be attained by a deterministic mechanism when the valuations satisfy a single-crossing condition. In contrast, when the valuations are the bidders' private information, we show that no finite bound on the social welfare can be achieved by any deterministic mechanism even under single-crossing. %, even if there are only two bidders. 
	Moreover, no randomized mechanism can guarantee better than $n$-approximation. We thus consider valuation functions that are submodular over signals (SOS), introduced in the context of combinatorial auctions in a recent breakthrough paper by Eden et al. [EC'19]. Our main result is an $O(\log^2 n)$-approximation randomized mechanism for buyers with private signals \emph{and} valuations under the SOS condition. We also give a tight $\Theta(k)$-approximation mechanism for the case each agent's valuation depends on at most $k$ other signals even for unknown $k$.
\end{abstract}

%\pagenumbering{gobble}
\newpage
%\pagenumbering{arabic}
%\clearpage
%\setcounter{page}{1}

\begin{comment}
\noindent To Do:
\begin{itemize}
    \item Prelims - unassigned
    \item Prelims: new payment identity \kgnote{Kira}
    \item Main proof: clean up \kgnote{Kira + Alon}
    \item Main proof: add intuition \kgnote{Kira + Alon}
    \item Deterministic LB -- use new payment identity \kgnote{Shuran} -- sec 5 (subsumed by 6?) and sec 6
    \item $O(k)$-apx for dependence $k$ \kgnote{Alon}
\end{itemize}

\end{comment}

\section{Introduction} \label{sec:intro}
Suppose a house is for sale.  Each prospective buyer researches the area, performs an inspection on the house, and does their best to learn how much the house is worth to them.  Some learn about the house's foundation, plumbing, roofing, and electrical condition.  Some learn about the neighborhood's public schools and community programs.  Others search for crime rates in the area.  Some look at upcoming construction or nearby highways.  While each prospective buyer has different information, which is their own private information, the amount that each buyer is willing to pay for the house depends on the information that \emph{all} of the prospective buyers have.  Similarly, one buyer may not personally know or care about public schools, but they may impact the \emph{resale value} of the house, and thus how much the buyer is willing to pay for the house.

%Consider a realty agency selling a house to a group of potential buyers. Each potential buyer inspects the house, and might gather some information about it --- one buyer might learn about the plumbing condition, another buyer might figure out the public schools' rankings in the school district, another buyer might find out about the neighborhood's criminal records, etc. While each potential buyer has a different piece of information, private to this buyer, they also care about all the mutual information that other buyers might gather, which might affect the amount of money they are willing to pay for the house. 

Settings where one buyer's value depends on another buyer's private information
%These kinds of situations where one buyer's private information might affect another buyer's value 
%\sz{valuation} 
are captured by the Nobel-prize-winning\footnote{This model was recently highlighted as one of the main achievements of Paul Milgrom and Robert Wilson in the scientific justification for their winning of the 2020 Sveriges Riksbank Prize in Economic Sciences in Memory of Alfred Nobel 2021~\cite{nobel2021considerations}.} \textit{Interdependent Value} (IDV) model introduced by Milgrom and Weber~\cite{milgrom1982theory} as a generalization of the  Common Value (mineral rights) model of Wilson~\cite{wilson1969communications}. In the IDV model, each buyer $i$ has 
a \emph{signal} $s_i\in \mathbb{R}^+$ known only to $i$ representing their piece of private information about the item, where this signal takes on a value from the signal space $S_i$.  Each buyer also has a \emph{valuation function} $v_i:\prod_j S_j\rightarrow \mathbb{R}^+$ that maps all $n$ of the buyers' private signals to a value (or willingness to pay) for the item for sale. Prior works considered the valuation functions to be 
%generally considered to be public information, known
public knowledge,
either to the seller~\cite{milgrom1982theory,RTC,Li17,CFK,EFFG,EFFGK,EdenFTZ21} and/or to other buyers~\cite{DM00,maskin1992,ItoP06,RobuPIJ13,EdenFTZ21}.  

While the public valuation assumption is used in all prior work, %as it is necessary for any results as elegant as attaining maximum welfare,
%is commonly used, and allows for beautiful results to be proven, 
we posit that it is simply unrealistic to assume that the valuation functions are public. Why should the seller be unaware of a buyer's private information, yet fully aware of the intricate manner in which the buyer takes into account all of the buyers' information when deciding their own actual value?
In contrast, our work, to the best of our knowledge, is \textbf{the first to study and provide approximation guarantees to social welfare in the IDV model when the valuation functions are the \emph{private information} of the buyers.}

When the valuations functions are public, a direct mechanism only needs to ask each agent to report their private signal $s_i$.  Then, it decides upon a winner and how much the winner pays by using the reported $s_i$'s and the publicly known valuation functions $v_i$'s. 
A key challenge of the interdependent literature is in eliciting agents' private information. Even when the valuations are public, a \textit{single-crossing} condition is required in order to obtain optimal welfare using an ex-post IC-IR mechanism for simple single parameter settings~\cite{ausubel1999generalized}.\footnote{Ex-post IC means that reporting the true signal is an ex-post Nash Equilibrium, and guarantees non-negative expected utility, see Definition~\ref{def:EPIC-IR}.} Only recently, approximation results have been shown for settings without single-crossing~\cite{EFFG,EFFGK,AmerTC21} (more on this in Section~\ref{sec:related}). Roughly speaking, single-crossing implies that if agent $i$'s signal changes, it will impact agent $i$'s valuation function more than any other agent---that is, of all agents with sensitivity to $i$'s signal, agent $i$'s sensitivity is the highest  
(see formally in Definition~\ref{def:sc}). 

But private valuations are much more challenging for welfare maximization: even for the single-item setting subject to single-crossing, we hit a wall with deterministic mechanisms, as demonstrated in the following example. 
\begin{example}[No finite bound for deterministic mechanisms under single-crossing] \label{example:sc-no-det}
	In this example, there are two agents, Alice and Bob, with signals from the $[0,1]$ interval (signal spaces $S_{A}=S_{B}=[0,1]$). 
	%Consider the following two cases: 
	
	%\noindent\textbf{Case 1.} 
	Suppose that Alice and Bob have the same valuation $v_A(s_A,s_B)=v_B(s_A,s_B)=s_A\cdot s_B$. Both Alice and Bob have signals equal to $1$, which implies they are both willing to pay $\$1$ for the item.
	Consider any deterministic mechanism $M$.  Without loss of generality, suppose that under these reports, the mechanism allocates to Alice. Then we imagine the following scenario.  
	
	%\noindent\textbf{Case 2.} 
	Alice's valuation is the same as before, while Bob has a valuation $v_B(s_A,s_B)=0.5+s_B$. Alice's signal is equal to $1$, while Bob's signal is now $0$.
	Now, Alice's value for the item is $0$, while Bob has a value of $0.5$. Therefore, to guarantee any finite approximation, the mechanism $M$ must allocate the item to Bob, and by individual rationality, Bob pays at most $0.5$.
	However, this implies that in the initial scenario when Bob has a valuation function $v_B(s_A,s_B)=s_A\cdot s_B$ and signal $s_B=1$, he would prefer to misreport both his valuation function and signal to this scenario, earning Bob the item with a payment of at most 0.5. Since his value is \emph{actually} $1$, he would earn positive (hence increased) utility. Hence, $M$ is not actually ex-post IC, and any deterministic mechanism must suffer an unbounded loss in welfare. Notice that 	this is true even though all valuations satisfy the single-crossing condition from Definition~\ref{def:sc} in this example.
	
	%In both cases, the valuations satisfy the single-crossing condition in Definition~\ref{def:sc}. Now assume that the seller does not know their valuations and decides the allocation based on their reported valuations and signals. In Case 1, in order to get any approximation, a deterministic mechanism must give the item to an agent. Let's assume w.l.o.g. the item is allocated to Alice. In Case 2, Alice's value for the item is $0$, while Bob has a value of $0.5$. Therefore, we must allocate the item to Bob, and by individual-rationality, Bob pays at most $0.5$. However, if this is the case, then in Case 1, Bob would rather misreport his value and signal to be the same as in Case 2. As a consequence, Bob will get the item, and pay at most 0.5. Since his value is $1$ in Case 1, he would incur a positive utility. Hence, there is no ex-post IC-IR mechanism that gives the item to Alice in Case 1, and to Bob in case 2, and any deterministic mechanism must suffer an unbounded loss in welfare.     
\end{example}
Thus, in stark contrast to when valuations are public information, no deterministic mechanism can guarantee \emph{any} finite approximation when buyer valuation functions are their own private information---even under the single-crossing condition, and  even when there are only two bidders. 

%Intuitively, the difficulty comes from the following. In the public valuations case,
When valuations are public, the auctioneer only needs to elicit signals.  If an agent lowers their report of their signal, they can decrease the value that the mechanism observes for other agents, and thus other agents' chances of winning.  However, under the single-crossing condition, agent $i$ is the most sensitive of all agents to $i$'s signal, so in fact their own valuation would decrease even more than others'.  This property enables the auctioneer to design an allocation rule for agent $i$ that is monotone non-decreasing in agent $i$'s reported signal and thus guarantee truthfulness.

However, when the valuations are unknown, the auctioneer elicits both signals and valuation functions, and this gives some intuition for why this setting is so much harder. % to approximate.  
The agents can now reduce their signal reports to lower others' observed values, and at the same time, ``compensate" for the loss in their own observed value by reporting a different ``higher" valuation function.  This is essentially what Bob can do to earn more utility in Example~\ref{example:sc-no-det}.  Therefore, when the valuations are private, the designer must guarantee that bidders cannot earn more utility by misreporting a combination of a valuation function and signal that will give the same observed value. %\shuran{Not sure if it's accurate or understandable, but the following sentence is a little bit confusing.} 
%they must not get a better allocation based on the specific signal and valuation function they reported. 
We formalize this notion in  Propositions~\ref{prop:epicir_suf} and \ref{prop:epicir_nec} by showing that the allocation can basically only depend on the reports of other agents' valuations $\v_{-i}$ and signals $\s_{-i}$, and the evaluation of bidder $i$'s valuation function on the entire actual signal profile $v_i(\s)$, but not the specific $v_i(\cdot)$ or $s_i$ that led to this value. Moreover, the allocation must be monotone non-decreasing in this value $v_i(\s)$.  

Using this characterization, we show that under single-crossing, even considering randomized allocation rules, the best mechanism cannot guarantee better than an $n$-approximation to the optimal welfare (see Proposition~\ref{prop:sc-n})---it might as well give the item to a random bidder. Therefore, unlike in the case of public valuations, in our case, single-crossing %does not seem to `buy' us anything
does not seem to simplify the problem. This leads us to consider
valuations that satisfy %\sz{another condition that guarantees a good approximation of optimal social welfare when the valuation functions are public,} 
the submodularity over signals (SOS) condition, introduced in the context of combinatorial auctions in the public valuation setting by Eden et al.~\cite{EFFGK}.

Submodularity over signals follows the intuition that a specific agent's information becomes less significant as we get stronger signals (more information) from other agents. This condition is satisfied by many well-known instances of interdependent valuations that appear in the literature such as mineral rights~\cite{wilson1969communications,RTC}, weighted-sum valuations~\cite{Myerson,klemperer1998auctions, RTC,EdenFTZ21,ChenEW}, maximum of signals~\cite{bulow2002prices,bergemann2020countering,ChenEW}, etc. When valuations are public information, Eden et al.~\cite{EFFGK} show that when valuations satisfy the SOS condition, there exists a mechanism that guarantees a 4-approximation to the optimal social welfare when the buyers are single-parameter, and there are downward-closed feasibility constraints on which agents can be simultaneously served (they do not require single-crossing). They also show that combined with a separability assumption, the SOS condition allows for a mechanism that can yield a 4-approximation for a very general setting of combinatorial auctions, overcoming well-known impossibility results~\cite{JM69,DM00}. Their mechanisms are inherently randomized, as can be shown in the following example from~\cite{EFFGK}. %for an auctioneer selling identical items, unit-demand buyers, and downward-closed constraints on which agents can be simultaneously sperved; they do not require single-crossing.  They also show that combined with a separability assumption, the SOS condition allows for a mechanism that can yield a 4-approximation for a very general setting of combinatorial auctions, overcoming well-known impossibility results~\cite{JM69,DM00}. Their mechanisms are inherently randomized, as can be shown in the following example from~\cite{EFFGK}.

\begin{example}[No finite bound for deterministic mechanisms under submodularity over signals \cite{EFFGK}] \label{example:sos-no-det}
	Consider the public valuation setting with two agents, Carl and Daphne, where only Carl holds  information regarding the item for sale. Carl's signal can be either low or high ($s_C\in\{0,1\}$). Carl's valuation function is $v_C(s_C) = 1+s_C$, while Daphne's valuation function is $v_D(s_C) = \alpha s_C$, where $\alpha\gg 1$. The valuations are SOS according to Definition~\ref{def:sos}. In order to have any approximation guarantee, a deterministic mechanism must allocate to Carl when his signal is low, as Daphne has value 0 when $s_C=0$, and must allocate to Daphne when Carl's signal becomes high. But again, this allocation rule cannot be supported by an ex-post IC-IR mechanism, since when Carl has a high signal, he'd rather misreport to $s_C=0$, get allocated, and pay at most 1 by individual rationality, granting him a strictly positive utility gain.
\end{example}

To summarize, deterministic mechanisms cannot guarantee any finite approximation to social welfare for valuations that satisfy single-crossing (Example~\ref{example:sc-no-det}) or SOS (Example~\ref{example:sos-no-det}), and randomized mechanisms cannot guarantee a better than $n$-approximation under single-crossing (Proposition~\ref{prop:sc-n}).  Naturally, we next consider randomized mechanisms for valuations that satisfy the SOS condition.  Our main result is as follows.

\vspace{0.1in}

\noindent\textbf{Main Result. } When valuations are the private information of bidders, and the valuations satisfy the submodularity over signals condition, 
there exists an ex-post IC-IR mechanism that gives an $O(\log^2 n)$-approximation to the optimal social welfare.

\vspace{0.1in}

We point out a few benefits of our result.  First, this setting and result is \emph{prior-free}---the guarantee holds for any signal profile $\s$ and valuations $\v$ without assuming that they are drawn from some Bayesian prior distribution.  Second, the mechanism can be implemented in polytime.  In fact, for each of the $n$ bidders, it only elicits each bidder's signal and makes $2n-1$ queries to each bidder's valuation function at different signal profiles.  It does \emph{not} ask the agents to disclose their entire valuation functions, supporting information minimization and making participation in the mechanism more desirable for privacy-aware buyers, as well as simplifying the mechanism. These properties make our result detail-free in arguably the strongest sense for the IDV setting, in line with Wilson's Doctrine~\cite{wilson1987game}, and as advocated by Dasgupta and Maskin~\cite{DM00}. 
 %maintaining both privacy for these functions and simplicity for the mechanism
More details about the main result can be found in Section~\ref{sec:sos}.

Moreover, we consider the problem of maximizing social welfare when the valuation functions are private and unconstrained by any assumptions (they only need to satisfy monotonicity in signals). For this setting, we give a mechanism whose performance is parameterized by the number of signals an agent's valuation might depend on. % \shuran{might care about?}. 
If every agent's valuation depends at most on $k$ other agents' signals, then our mechanism achieves a $2(k+1)$ approximation (Theorem~\ref{thm:k-dep}). This is tight up to constant factors, as 
\begin{itemize}
	\item for every $k$ there exists an instance with \textit{public} valuations that depend on at most $k$ other agents' signals, and no mechanism can get better than $k+1$-approximation; 
	\item  for every $k$ there exists an instance with \textit{private} valuations that (1) depend on at most $k$ other agents' signals and (2) satisfy single-crossing, and no mechanism can get better than $k+1$-approximation. 
\end{itemize}

\subsection{Road Map} In Section~\ref{sec:prelims} we formally present our model, where in Section~\ref{sec:truthful-char} we provide the truthfulness characterization of our setting, and show the implication regarding randomized mechanisms under SC valuations. In Section~\ref{sec:sos} we present our main result in devising an $O(\log^2 n)$-approximation mechanism under SOS valuations. Once we construct a mechanism according to our truthfulness characterization, the remaining challenges are feasibility and proving that our approximation is achieved.  First, we discuss high-level intuition for our methods.  Then, in subsection~\ref{sec:feasible} we formally prove feasibility, and in subsection~\ref{sec:approx} we formally prove the approximation. Finally, in Section~\ref{sec:k-depend} we devise a mechanism that gives an $O(k)$-approximation to the optimal welfare when each agent's valuation function depends on at most $k$ other agents, and we show it's tight up to a constant factor.

%\begin{itemize}
%	\item Interdependent values are important. Nobel prize.
%	\item private signals, public valuations.
%	\item single item is solved under single-crossing.
%	\item Assumes public valuations.
%	\item A recent line of works tries to relax stringent conditions under which positive results can be attained.
%	\item *ALL* works assume public valuations.
%	\item private valuations make things crazy hard. Insert example.
%	\item give intuition for what goes wrong.
%	\item Our positive results.
%1\end{itemize}

\subsection{Related Work} \label{sec:related}
\paragraph{Valuations unknown to the seller.} 
%\vspace{.3cm}
%\textbf{Valuations unknown to the seller.} 
The most closely related work to our work is the seminal (and beautiful) paper of Dasgupta and Maskin~\cite{DM00}. Dasgupta and Maskin design auctions where the seller is oblivious of the bidders' valuation functions, and bidders do not report their valuation function to the seller, but rather report a contingent bidding function in the form of ``if the other bidder bids $\$x$, I will bid $\$y$." They show these auctions achieve full efficiency in equilibrium under various crossing-type conditions in a wide variety of settings. However, in order for the bidders to reach an equilibrium, they must be aware of one another's valuation functions, or at least, a non-trivial summary of them, violating the premise of private valuations. We, on the other hand, design mechanisms for fully private valuations, and achieve stronger incentive guarantees. Follow-ups on the contingent bidding model appear in~\cite{RobuPIJ13,ItoP06}.

%\paragraph{Submodularity-type conditions.}
\textbf{Submodularity-type conditions.}
Prior to Eden et al.~\cite{EFFGK}, a weaker form of submodularity over signals appeared in Chawla, Fu and Karlin~\cite{CFK}, along with single-crossing, in the context of revenue maximization. It was also used in Eden et al.~\cite{EFFG} for a more general reduction from welfare maximization to revenue maximization, and in Eden et al.~\cite{EdenFTZ21} to help bound the  price of anarchy of simultaneous auctions in the IDV model. % \kgnote{Add \cite{EFFGK}?}

%\paragraph{Removing assumptions.} 
\textbf{Removing assumptions.} 
A recent line of work in the EconCS literature study removing assumptions in the IDV model. Roughgarden and Talgam-Cohen~\cite{RTC} study prior-independent mechanisms for revenue maximization under matroid feasibility constraints, and Li~\cite{Li17} shows that the VCG mechanism with monopoly reserves gets an $1/e$-fraction of the optimal revenue. Chawla, Fu and Karlin~\cite{CFK} study approximately-optimal mechanisms for revenue maximization under minimal assumptions. Eden et al.~\cite{EFFG} study welfare and revenue maximization in the absence of a strict single-crossing condition.  Eden et al.~\cite{EFFGK,EdenFTZ21} bypass the impossibility results of Dasgupta and Maskin~\cite{DM00} and Jehiel and Moldovanu~\cite{JM69} in combinatorial auctions settings by considering approximation to the optimal welfare instead of full efficiency. %\kgnote{last sentence vague?}

%\subsection{Road Map}
%\kgnote{Add a road map of the paper here.}

\section{Preliminaries} \label{sec:prelims}
We consider the setting where a single item is being auctioned off and bidders have \emph{private} interdependent valuations.  That is, each of the $n$ bidders has a private signal $s_1, \dots, s_n$, with $s_i$ in a bounded signal space $S_i$. We assume signals are bounded and in $[0,1]$.\footnote{Assuming signals are in $[0,1]$ is without loss for bounded signals, as one can redefine the valuation functions such that this is the case.} Additionally, bidder $i$ has a valuation function $v_i:[0,1]^n\rightarrow\mathbb{R}^+$ which (1) maps signals of the bidders to a real value and (2) is a part of $i$'s \textit{private} information. The function $v_i$ is monotone non-decreasing in all signals, and strictly increasing in $i$'s signal. Let $\s$ ($\v$) denote the vector of all signals (valuations), and $\s_{-i}$ ($\v_{-i}$) denote the vector of all signals (valuations) without agent $i$'s signal (valuation). Let $S$, $V_i$, and $V$ denote the spaces of $\s$, $v_i$, and $\v$ respectively. For an integer $\ell$, let $[\ell]=\{1,\ldots,\ell\}$. 
%(1) is his own \emph{private} information and (2) depends on the signals of all bidders. 
%%% DELETED -- BAYESIAN
%Assume the signals are jointly drawn from distribution $F$ with density function $f$, the valuations are jointly drawn from distribution $G$ with density function $g$.

Our goal is to design a mechanism that maximizes social welfare. A mechanism $M = (x,p)$ consists of an allocation rule $x$ and a payment rule $p$. The mechanism first asks each bidder to submit a bid---without loss of generality, we consider direct mechanisms in which the bidders report a bids for their signals $\tilde{s_i}$ and valuation functions $\tilde{v_i}$---and then the mechanism allocates the item with probability $x_i(\tilde{\v}, \tilde{\s})$ to bidder $i$ and charges them $p_i(\tilde{\v}, \tilde{\s})$.  %\footnote{Our mechanisms can actually be implemented using queries to the valuation functions, and do not require agents to fully disclose their valuation functions.} 
An agent's utility from a bid profile $\tilde{\v},\tilde{\s}$ when $i$'s real valuation is $v_i$, and the real signal profile is $\s$ is $x_i(\tilde{\v},\tilde{\s})v_i(\s)-p_i(\tilde{\v},\tilde{\s})$.

%Two solution concepts are considered.
We use the following solution concept, which is the strongest for the interdependent setting.
\begin{definition}[Ex-post IC and ex-post IR] \label{def:EPIC-IR}
	A mechanism is ex-post IC if the bidder will not regret truthfully reporting after seeing other bidders' true signals, formally, for all true profiles $\v,\s$, bidders $i$, and reports $v'_i,s'_i$,
	$$
	v_i(\s) x_i(\s, \v) - p_i(\s, \v) \ge v_i(\s) x_i(\s_{-i}, s_i', \v_{-i}, v_i') - p_i(\s_{-i}, s_i', \v_{-i}, v_i').
	$$
	Similarly, a mechanism is ex-post IR if 
	$
	v_i(\s) x_i(\s, \v) - p_i(\s, \v) \ge 0.
	$
	
	We use EPIC-IR to denote a mechanism that is both ex-post IC and ex-post IR.
\end{definition}

\subsection{Standard Conditions on Valuations}
In this paper we focus on valuations that are \textit{submodular over signals} (SOS) introduced by~\cite{EFFGK} in the context of combinatorial auctions. 
\begin{definition}[Submodularity over signals]
	A valuation $v:S\rightarrow\mathbb{R}^+$ is said to be \emph{submodular over signals} if for every $i$, $\s_{-i}$ and $\s'_{-i}$ such that $\s_{-i}$ is coordinate-wise %smaller or equal than 
	smaller than or equal to
	$\s'_{-i}$, every $s_i$ and $\delta >0$, 
	\begin{eqnarray}
		v(s_i+\delta,\s_{-i}) - v(s_i,\s_{-i}) \ \geq \ v(s_i+\delta,\s'_{-i}) - v(s_i,\s'_{-i}). 
	\label{def:sos} \end{eqnarray} 	
\end{definition}
Roughly speaking, this condition comes to represent situations where valuations have diminishing marginal returns in signals.

We also discuss a widely-used condition called \textit{a single-crossing} condition.
\begin{definition}[Single-crossing] \label{def:sc}
	A valuation profile $\v$ is said to satisfy the \emph{single-crossing} (SC) condition if for every $i,j$, every $s_i,\s_{-i}$ and every $\delta>0$,
	\begin{eqnarray}
		v_i(\s_i+\delta,\s_{-i}) - v_i(\s_i,\s_{-i}) \ \ge \ v_j(\s_i+\delta,\s_{-i}) - v_j(\s_i,\s_{-i}).
	\end{eqnarray}
\end{definition}
Roughly speaking, this condition represents situations where $i$'s signal affects $i$'s valuation more than the valuation of any other agent $j$. The importance of this condition is comes from the fact that this condition (or similar relaxations of it) guarantees the existence of socially optimal mechanisms in the case where valuations are public, and only signals are private, in the context of single-item auctions~\cite{maskin1992,ausubel1999generalized} and multi-item auctions under matroid feasibility constraints~\cite{ausubel2004efficient,RTC,Li17}.

%%%%% BAYESIAN
\begin{comment}
\begin{definition}[Bayesian IC and interim IR] A mechanism is Bayesian IC if 
	$$
	\E_{s_{-i}, v_{-i}}\left[v_i(\s) x_i(\s, \v) - p_i(\s, \v) \ge v_i(\s)\right] \ge \E_{s_{-i}, v_{-i}}\left[ x_i(\s_{-i}, s_i', \v_{-i}, v_i') - p_i(\s_{-i}, s_i', \v_{-i}, v_i')\right].
	$$
	A mechanism is interim IR if 
	$$
	\E_{s_{-i}, v_{-i}}\left[v_i(\s) x_i(\s, \v) - p_i(\s, \v)\right] \ge 0.
	$$
\end{definition}
\end{comment}

\subsection{Truthfulness Characterization and Implications.} \label{sec:truthful-char}
The following propositions, whose proofs are deferred to Section~\ref{sec:epicir-proof}, give strong restrictions on the form of any EPIC-IR mechanism in our setting.  While in hindsight it makes sense that we can only use bidder $i$'s real-numbered value to determine his allocation---and not his signal or valuation function, each of which he could manipulate to impact other agents then compensate for---the steps to derive this are in fact somewhat more complicated than in Myerson~\cite{Myerson} and the public valuation IDV analogue from Roughgarden Talgam-Cohen~\cite{RTC}.

We first give a sufficient condition for an allocation rule to be implementable.
\begin{proposition}[Sufficient conditions]\label{prop:epicir_suf}
	Consider a set of functions $\{x_i: V\times S \to \mathbb{R}\}$. %A mechanism that for bidder reports $\v,\s$, allocates each bidder with probability $x_i(\v_{-i},\s_{-i},v_i(\s))$ can be implemented in an EPIC-IR way. \alon{(IFF??)}
	There exists a payment rule $\{p_i: V\times S \to \mathbb{R}\}$ so that $(x,p)$ is EPIC-IR %A mechanism is EPIC-IR 
	if for every $i$, $\v_{-i},\s_{-i}$, 
	\begin{enumerate}[label=(\roman*)]
	    \item If $v_i(\s_{-i}, s_i) = \hat v_i(\s_{-i}, \hat s_i)$ then $x_i(\v_{-i}, v_i,\s_{-i}, s_i) = x_i(\v_{-i},\hat v_i,\s_{-i}, \hat s_i)$.  That is, we may as well write $x_i(\v_{-i},\s_{-i}, v_i(\s))$ as $x_i$ is dependent only on these three parameters.
	   \item $x_i(\v_{-i},\s_{-i},\cdot)$ is monotone non-decreasing.
	\end{enumerate}
	%Consider a set of functions $\{x_i: V_{-i}\times S_{-i}\times \mathbb{R}^+\mapsto [0,1]\}_{i\in [n]}$ which for every $i$, $\v_{-i},\s_{-i}$, $x_i(\v_{-i},\s_{-i},\cdot)$ is monotone non-decreasing. A mechanism is EPIC-IR if and only if 

\end{proposition}

We further show that this sufficient condition is almost necessary.
\begin{proposition}[Necessary conditions]\label{prop:epicir_nec}
	Consider a set of functions $\{x_i: V\times S \to \mathbb{R}\}$. %A mechanism that for bidder reports $\v,\s$, allocates each bidder with probability $x_i(\v_{-i},\s_{-i},v_i(\s))$ can be implemented in an EPIC-IR way. \alon{(IFF??)}
	There exists a payment rule $\{p_i: V\times S \to \mathbb{R}\}$ so that $(x,p)$ is EPIC-IR only if for every $i$, $\v_{-i},\s_{-i}$, 
	\begin{enumerate}
		\item For every $v_i,s_i$ $\hat{v}_i,\hat{s}_i$ such that $v_i(s_i,\s_{-i})>\hat{v}_i(\hat{s}_i, \s_{-i})$, $x_i(\v_{-i},\s_{-i},v_i,s_i)\geq x_i(\v_{-i},\s_{-i},\hat{v}_i,\hat{s}_i)$.
	    \item %If $v_i(\s_{-i}, s_i) = \hat v_i(\s_{-i}, \hat s_i)$ then $x_i(\v_{-i}, v_i,\s_{-i}, s_i) = x_i(\v_{-i},\hat v_i,\s_{-i}, \hat s_i)$.  That is, we may as well write $x_i(\v_{-i},\s_{-i}, v_i(\s))$ as $x_i$ is dependent only on these three parameters.
	    $x_i(\v_{-i}, v_i,\s_{-i}, s_i)$ almost only depends on $v_i(\s)$ for fixed $\v_{-i}, \s_{-i}$. More specifically, define function $x_{v_i}: \mathbb{R}^+ \to \mathbb{R}^+$ with $$x_{v_i}(y) = x_i(\v_{-i}, v_i, s_{-i}, v_i^{-1}(y)),$$ which maps a possible valuation $y$ of bidder $i$ to the corresponding allocation $x_i$ when bidder $i$'s valuation function is $v_i$. We assume that $x_{v_i}$ is an integrable function on $[v_i(0, \s_{-i}), +\infty)$ for all $v_i \in V_i$.  Then for any $v_i, \hat v_i \in V_i$, defining $L = \max\{v_i(0, \s_{-i}),\hat v_i(0,\s_{-i})\}$,
    \begin{align*}
    x_{v_i}(\cdot) = x_{\hat v_i} (\cdot) \text{ almost everywhere on } [L, +\infty). 	
    \end{align*}
%	   \item $x_i(\v_{-i},\s_{-i},\cdot)$ is monotone non-decreasing.
	\end{enumerate}
	%Consider a set of functions $\{x_i: V_{-i}\times S_{-i}\times \mathbb{R}^+\mapsto [0,1]\}_{i\in [n]}$ which for every $i$, $\v_{-i},\s_{-i}$, $x_i(\v_{-i},\s_{-i},\cdot)$ is monotone non-decreasing. A mechanism is EPIC-IR if and only if 

\end{proposition} 

To prove the sufficient direction of the above proposition, we use the following payment rule, which along with an allocation that satisfies conditions (i) and (ii) above, results in an EPIC-IR mechanism. 
\begin{eqnarray}
	p_i(\v,\s) = x_i(\v_{-i},\s_{-i},v_i(\s))v_i(\s) - \int_{0}^{v_i(\s)}x_i(\v_{-i},\s_{-i},t) dt.\label{eq:payment}
\end{eqnarray}
We refer to this payment rule in our mechanisms' description.

An immediate implication of the above proposition is the strong impossibility for single-crossing valuations which, unlike in the public valuation setting, do not give any non-trivial guarantee here. %We defer the proof to Appendix~\ref{sec:sc-n-proof}.
\begin{proposition}
	No randomized EPIC-IR mechanism can do better than an $n$-approximation under the single-crossing condition when valuation functions are private. \label{prop:sc-n}
\end{proposition}
\begin{proof}
	Consider two cases:\footnote{For readability, the in the following example, the valuations are not strictly increasing in the agents' own signals as defined above. This can be modified easily by adding an additive $\epsilon s_i$ to each agent $i$'s valuation.}
	\begin{itemize}
		\item \textbf{Case 1.} $n$ agents with identical valuation $v_i(\s)=\prod_{j=1}^n s_j$,  and $\s=\mathbf{1}$. The valuations satisfy single-crossing. There must be an agent who's allocation probability is at most $1/n$. Assume without loss it is agent $1$.
		\item \textbf{Case 2.} $v'_1=0.5+s_1$, and $v_2,\ldots, v_n$ are as before. Note that single-crossing is still satisfied. The signals are $s'_1=0$, $\s_{-1}=\textbf{1}$. 
	\end{itemize} 
	Since in both cases $\s_{-1}$ and $\v_{-1}$ are the same, and $v'_1(s'_1,\s_{-1})< v_1(s_1,\s_{-1})$, by Proposition~\ref{prop:epicir_nec}, in Case 2, agent $1$'s allocation probability is at most $1/n$. Since agent 1 is the only agent with non-zero value, this implies any EPIC-IR mechanism can guarantee at most $1/n$ of the optimal welfare.
\end{proof}

%\section{Single Crossing is not Enough}
%\input{sc-impossibility}
%\input{impossibility}
\section{$O(\log^2 n)$-approximation under Submodularity}\label{sec:sos}

In this section, we consider the case where every $v_i(\cdot)$ is submodular over signals (see Definition~\ref{def:sos}). We assume without loss of generality that signals are in the $[0,1]$ interval.\footnote{One can always rescale signals and adjust valuations so that this is the case. For instance, with valuation $v_i(\cdot)$ and signal space $S_i = [\underline{s}_i,\overline{s}_i]$, we can define a new valuation function $\tilde{v}_i(\cdot)$ based on $v_i(\cdot)$ that uses signal space $\tilde{S}_i = [0,1]$: for any $\s_{-i}$, $\tilde{v}_i(s_i,\s_{-i}) = v_i(\underline{s}_i+s_i(\overline{s}_i-\underline{s}_i),\s_{-i})$.} In this section, we devise an EPIC-IR mechanism that gives $O(\log^2 n)$-approximation to the optimal social welfare. We first provide intuition for our mechanism.

\paragraph{High-level intuition.} Devising EPIC-IR mechanisms for our setting is very tricky. On one hand, the private type of each agent is \emph{multi-dimensional}, as it contains the space of all submodular valuation functions. On the other, our truthfulness characterization (Propositions~\ref{prop:epicir_suf}, \ref{prop:epicir_nec}) implies that the mechanism can basically only take into account a \emph{single-dimensional summary} of $i$'s private type, $v_i(\s)$, in order to determine agent $i$'s allocation probability.  This also implies that $v_j(\s)$ cannot be used to determine agent $i$'s allocation probability for any $j \neq i$, as $\s$ uses $s_i$.  Together, these facts rule out many natural mechanisms as not EPIC-IR: e.g., allocating to the highest-valued agent (as shown in Example~\ref{example:sc-no-det}), or posting an anonymous price for the item (determined by using a sample of the agents) and then sequentially approaching the buyers to offer the item at this price.\footnote{Consider an agent that comes later in the sequential anonymous price mechanism. Such an agent can lower their reported signal to prevent an earlier agent from having a high enough value to purchase the item, and then compensate by ``increasing" their reported valuation function.}  It also rules out mechanisms that allocate to the bidder with the highest \emph{proxy} value, i.e.~using $v_i(s_A, 0_B)$ % where the mechanism determines $\hat v_i$ and the sets $A \sqcup B$ 
for some partition $A \sqcup B$ of the agents determined by the mechanism (as used in~\cite{EFFGK}).\footnote{The mechanism cannot use the real value/signals to charge the agent, as the agent might misreport those to get a lower payment. If the mechanism uses the proxy value, then if the proxy value is smaller than an agent's actual value, an agent might gain by misreporting a higher valuation, and increase the proxy value.} To make things worse, by Example~\ref{example:sos-no-det}, we must only consider randomized mechanisms.  

%These together 

%As a bidder $i$'s allocation is a function of $\s_{-i}$, $\v_{-i}$, and $v_i(s)$, a bidder can impact their own allocation probability only with their report of $v_i(s)$.  We do not consider sequential mechanisms due to the above reasoning. Moreover, 
%so each bidder must be set a (inherently randomized by Example~\ref{example:sos-no-det}) allocation probability exists simultaneously and these together must be feasible.

Remember that by Propositions~\ref{prop:epicir_suf}, \ref{prop:epicir_nec}, we consider allocation functions such that bidder $i$'s allocation is a function of $\s_{-i}$, $\v_{-i}$, and $v_i(\s)$. When setting $i$'s allocation function, the closer $i$ is to the highest-valued agent, that is, the fewer agents who have values above $i$, the higher probability of allocation we would like to give $i$. However, when determining $i$'s allocation, we are not allowed to use $s_i$, so we cannot know who the agents are with values higher than $i$.  Instead, we can only $\emph{estimate}$ which agents have a higher value than $i$ without using $s_i$ or $v_i(\cdot)$.  We do this by looking at the agents $j$ such that $v_j(1_i,\s_{-i}) > v_i(\s)$. % \sz{so that an agent cannot misreport to change our estimate}. \shuran{I think we probably should explain that we cannot know this due to truthfulness constraint? It is not very clear here why we cannot know but can estimate using $v_j(1_i,\s_{-i})$} 
This is roughly the same set of agents as $D_i$, defined in Equation~\eqref{eq:di} below. %\shuran{I moved the following sentences to the beginning: Intuitively, we want to get an agent $i$ a higher allocation probability if there are less potential bidders that might have a larger value than $i$. This is captured by Equation~\eqref{eq:allocation} below.}
 We set the allocation probability $x_i$ to be inversely proportional to the size of this set $D_i$ (Equation~\eqref{eq:allocation}). Ideally, this set should be small for the \textit{highest-valued} agent $i^*$ so that $i^*$ gets allocated with high probability. However, this might not be the case, and the set $D_{i^*}$ might be as large as $\Omega(n)$. To deal with such a case, we want to restrict the set $D_{i^*}$ to contain only agents that are close in value to agent $i^*$ (as depicted by the constraint that $v_j(0_i,\s_{-i})\geq (1-1/\log n)v_i(\s)$ in the definition of $D_i$).  Now, if these agents win the item instead, we know they're close enough in value to $i^*$ for a good approximation, so we just need to reason that these agents in $D_{i^*}$ also have large enough allocation probability to compensate for $i^*$'s small probability in this case. In our analysis, we show that the combined probability of agents not ``too far" from $i^*$'s value is large enough, hence guaranteeing a good approximation to welfare. This is shown in Section~\ref{sec:approx}.

% [took out more formal intuition and moved to 3.2]

In our setting, since we jointly define the allocation functions, proving feasibility is no less tricky. % to prove that our mechanism assigned a feasible allocation probability. 
Simply setting the allocation probability of agent $i$ to be $1/(|D_i|+1)$ might result in an infeasible mechanism. 
In Section~\ref{sec:feasible} we show how to define the allocation rule in order to ensure feasibility. At a very high-level, we show that the allocation rule is closely related to the coloring number of an adequately defined directed graph. This coloring number can then be bounded using submodularity properties of the valuation functions. This is the only place where submodularity is used in the proof.

\paragraph{The mechanism.}
%Consider the case where every $v_i$ is SOS, and assume WLOG $\s\in [0,1]^n$.\footnote{One can always rescale signals and adjust valuations so this is the case.} 
For each bidder $i$, we try to find the bidders who have higher valuations than $i$ without using $i$'s bid. %Without looking at $i$'s bid, we can know that a bidder $j$'s valuation must lie in range $[v_j(0_i, \s_{-i}), v_j(1_i,\s_{-i})],$ assuming that everyone reports truthfully. 
Define 
\begin{eqnarray}
	D_i(\tilde{\v},\tilde{\s}) \ =\  \left\{j \ : \ \tilde{v}_j(1_i,\tilde{\s}_{-i}) \geq \tilde{v}_i(\tilde{\s}) \ \wedge \ \tilde{v}_j(0_i, \tilde{\s}_{-i})\geq \left(1-\frac{1}{\log n}\right)\tilde{v}_i(\tilde{\s})\right\} =\  D_i(\tilde{\v}_{-i},\,\tilde{\s}_{-i}, \,\tilde{v}_i(\tilde{\s})) \ \label{eq:di}
\end{eqnarray}
	to be the bidders $j$ who have the highest possible valuation $\tilde{v}_j(1_i,\tilde{\s}_{-i})$ no less than $\tilde{v}_i(\tilde{\s})$ and the lowest possible valuation $\tilde{v}_j(0_i, \tilde{\s}_{-i})$ no less than $\left(1-1/\log n\right)\tilde{v}_i(\tilde{\s})$. The set $D_i$ is depicted in Figure~\ref{fig:di}.

Consider the following mechanism:
\begin{enumerate}
	\item Elicit bids $(\tilde{v}_i,\tilde{s}_i)$ for every agent $i$.\footnote{We later show how the mechanism can be implemented using only $n(2n-1)$ calls to the valuation oracles of the agents.}
	%\item For every $i$, let $$D_i(\tilde{\v},\tilde{\s}) \ =\  D_i(\tilde{\v}_{-i},\tilde{\s}_{-i}, \tilde{v}_i(\tilde{\s})) \ =\  \left\{j \ : \ \tilde{v}_j(1_i,\tilde{\s}_{-i}) \geq \tilde{v}_i(\tilde{\s}) \wedge \tilde{v}_j(0_i, \tilde{\s}_{-i})\geq \left(1-1/\log n\right)\tilde{v}_i(\tilde{\s})\right\}.$$
	\item Each agent $i$ is allocated with probability 
	\begin{eqnarray}
		x_i(\tilde{\v},\tilde{\s})\ =\ \frac{1}{\chi}
	\frac{1}{\ln n + 1}\frac{1}{|D_i(\tilde{\v}_{-i},\tilde{\s}_{-i}, \tilde{v}_i(\tilde{\s}))| + 1}=\ x_i(\tilde{\v}_{-i},\,\tilde{\s}_{-i}, \,\tilde{v}_i(\tilde{\s})) \ ,\label{eq:allocation}
\end{eqnarray}
	where $\chi\geq 1$ is a constant. We later show that setting $\chi=2\log n+1$ ensures feasibility.  %to be determined later. %\shuran{should we directly set $\chi$ equal to $2\log n+1$ because the agents can misreport to change $\chi$?} 
	Charge payments according to Equation~\eqref{eq:payment}.
\end{enumerate}

%Consider the following mechanism: bidders submit valuations $\tilde{\v}$ and signals $\tilde{\s}$. Bidder $i$ is allocated with probability 
%\begin{eqnarray}
%	\alpha_i = \frac{1}{\chi}
%	\frac{1}{\ln n + 1}\frac{1}{|D_i(\tilde{\v},\tilde{\s})| + 1},\label{eq:allocation}
%\end{eqnarray}
%where $\chi\geq 1$ is to be determined later.  

We show that the mechanism satisfies the EPIC-IR conditions in Proposition~\ref{prop:epicir_suf}.  %\kgnote{Fill in props, $\chi$?}

\begin{lemma} 
	The mechanism can be implemented in an ex-post IC-IR manner. \label{lem:main-epicir}
\end{lemma}
\begin{proof}
	Fix $\tilde{\v}_{-i},\tilde{\s}_{-i}$. Notice that as $\tilde{v}_i(\tilde{s})$ increases, the size of the set $D_i(\tilde{\v}_{-i},\tilde{\s}_{-i}, \tilde{v}_i(\tilde{\s}))$ can only decrease. Since $x_i$ is inversely proportioned to $|D_i(\tilde{\v}_{-i},\tilde{\s}_{-i}, \tilde{v}_i(\tilde{\s}))|$, we get that $x_i(\tilde{\v}_{-i},\tilde{\s}_{-i}, \cdot)$ is monotonically non-decreasing in the last parameter, satisfying the properties of Proposition~\ref{prop:epicir_suf}, which yields the lemma.
%	Let us define the set $$\hat{D}_i(\v_{-i},\s_{-i}, x) = \left\{j \ : \ v_j(1_i,\s_{-i}) \geq x \wedge v_j(0_i, \s_{-i})\geq \left(1-1/\log n\right)x\right\}.$$ Notice that $\hat{D}_i(\v_{-i},\s_{-i}, v_i(\s)) = D_i(\v,\s)$. Moreover, notice that for a fixed $i,\v_{-i},\s_{-i}$, as $x$ increases, $|\hat{D}_i(\v_{-i},\s_{-i}, x)|$ can only decrease.
%	Let 
%	$$x_i(\v_{-i},\s_{-i},x) = \frac{1}{\chi}
%	\frac{1}{\ln n + 1}\frac{1}{|\hat{D}_i(\v_{-i},\s_{-i}, x)| + 1}.$$
%	Since for a fixed $i,\v_{-i},\s_{-i}$, $|\hat{D}_i(\v_{-i},\s_{-i}, \cdot)|$ is non-increasing, $x_i(\v_{-i},\s_{-i},\cdot)$ is non-decreasing, satisfying the conditions of Proposition~\ref{prop:epicir}, and can be implemented in an EPIC-IR manner. Moreover, for $\tilde{\v},\tilde{\s}$, $$x_i(\tilde{\v}_{-i},\tilde{\s}_{-i},\tilde{v}_i(\tilde{\s})) \ =\  \frac{1}{\chi}
%	\frac{1}{\ln n + 1}\frac{1}{|\hat{D}_i(\tilde{\v}_{-i},\tilde{\s}_{-i}, \tilde{v}_i(\tilde{\s}))| + 1} \ =\  \frac{1}{\chi}
%	\frac{1}{\ln n + 1}\frac{1}{|D_i(\tilde{\v},\tilde{\s})| + 1} = \alpha_i.$$
%	Therefore, $x_i$ implements the allocation rule of the mechanism, which concludes the proof.
\end{proof}

Since our mechanism is ex-post IC-IR, we assume that every bidder reports truthfully in the rest of the section. 
For ease of notation, we fix truthful bids $(\v,\s)$, and drop $\v$ and $\s$ from the notation when clear from context. The challenge in proving the mechanism performs well is in setting $\chi$ just right so that the mechanism is both feasible (the sum of allocation probabilities is less than $1$) \textit{and} has a good approximation guarantee. In Section~\ref{sec:feasible} we tie the value of $\chi$ to the coloring number of an appropriately defined graph. In Section~\ref{sec:approx}, we prove the desired approximation guarantee.

\begin{figure}[h!]
	\centering
	\begin{subfigure}{0.45\textwidth} 
		\centering
		\includegraphics[width=\textwidth]{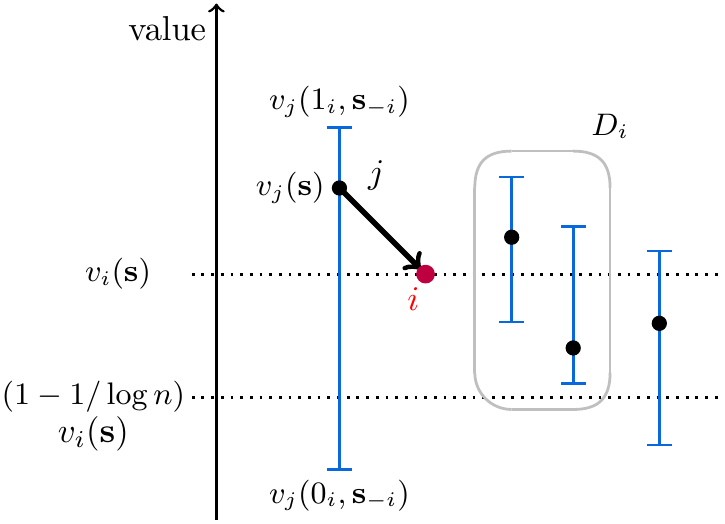}	
		\caption{Edge $j\to i$ and set $D_i$.}
		\label{fig:di}
	\end{subfigure}
	~\hfill
	\begin{subfigure}{0.45\textwidth} 
		\centering
		\includegraphics[width=\textwidth]{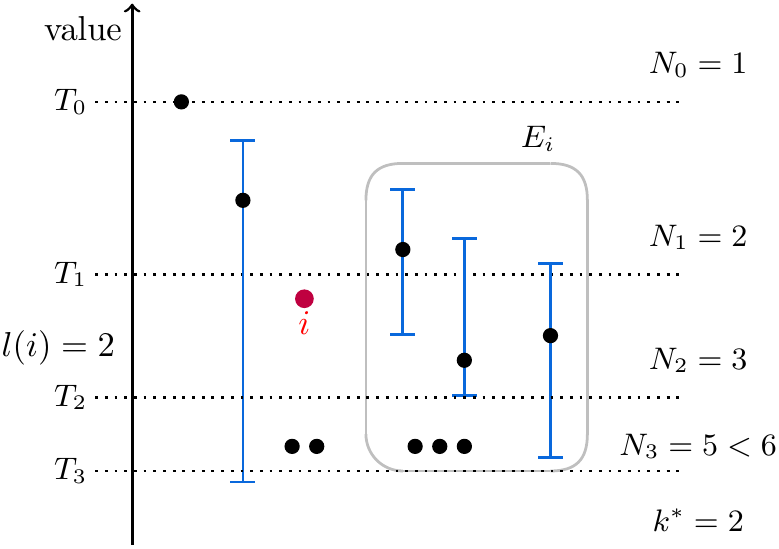}
		\caption{Thresholds, set $E_i$, $N_k$, and $k^*$.}
		\label{fig:ei}
	\end{subfigure}
	\caption{Illustrations of the definitions used in Section~\ref{sec:sos} with respect to agent $i$ (pink). Each point in the graphs represents a bidder and the $y$-axis represents the bidder's valuation. Each blue segment shows the range of a bidder $j$'s possible valuation for different values of (an unknown) $s_i$, with the highest point equal to $v_j(1_i, \s_{-i})$ and the lowest point equal to $v_j(0_i, \s_{-i})$. Figure (a) shows a possible $D_i$ (gray box, defined in~\eqref{eq:di}) and a possible edge $j\to i$ in graph $G_{\v,\s}$ (Definition~\ref{def:G_vs}) for a ``missed" $j \not \in D_i$. Figure (b) illustrates the thresholds, the set $E_i$, and the quantities $N_k$ and $k^*$ defined in Section~\ref{sec:approx}. }
\end{figure}

\subsection{Feasibility} \label{sec:feasible}
In this section, we prove that our allocation rule is feasible, i.e., $\sum_{i=1}^n x_i \le 1$. 
We first note that if $D_i$ contains all the bidders with values at least as high as $i$'s value (when everyone reports truthfully), then feasibility is immediate. The argument for this is: if we rename the agents according to the order of values ($1$ being the highest-valued agent), then the renamed $D_i$ contains all agents $j < i$, hence $|D_i| \geq i-1$, implying $x_i\leq \frac{1}{\ln n +1 }\frac{1}{i}$, and thus, 
\begin{align}
\sum_i x_i  \leq \frac{1}{\ln n +1 }\sum_i\frac{1}{i}\leq 1. \label{eqn:e_bound}	
\end{align}
However, $D_i$ does \emph{not} necessarily contain all agents $j$ with $v_j(\s)\geq v_i(\s)$, as it may be that  $v_j(0_i,\s_{-i}) < \left(1-1/\log n\right)v_i(\s)$, and thus $j \not \in D_i$.
In order to handle such ``missed" agents, we introduce the following graph.

\begin{definition} \label{def:G_vs}
	Let $G_{\v,\s}=([n],A)$ be the following directed graph:
	\begin{itemize}
		\item The vertices are the set of $n$ agents.
		\item For two vertices $j\neq i$ representing two agents, there is a directed edge $j\rightarrow i$ if $v_j(\s)\geq v_i(\s)$ and $v_j(0_i,\s_{-i}) < \left(1-1/\log n\right)v_i(\s)$; that is, if $j$ has a higher value than $i$, but is not counted in the set $D_i$ that determines $i$'s allocation probability.
	\end{itemize} %, where , and there is a directed edge $\langle j, i\rangle$ between agent $j$ and agent $i$ if $j\rightarrow i$.
\end{definition}
Figure~\ref{fig:di} illustrates a case where there is a directed edge between agent $j$ and agent $i$ in $G_{\v,\s}$.

The main idea of our proof is that if we divide the vertices into independent sets, then the allocation probability within each independent set can be bounded as in~\eqref{eqn:e_bound} because there's no ``missed'' agents within this independent set. Therefore, we can relate the value $\chi$ in Equation~\eqref{eq:allocation} to the coloring number of $G_{\v,\s}$.

\begin{lemma}
	Setting $\chi$ to be greater than or equal to $\chi(G_{\v,\s})$ (the coloring number of $G_{\v,\s}$) in Equation~\eqref{eq:allocation} ensures that the mechanism is feasible. \label{eq:chi-fesaible}
\end{lemma}
\begin{proof}
	Let $\chi = \chi(G_{\v,\s})$, consider a coloring of $G_{\v,\s}$ with $\chi$ colors, and for every color $c\in [\chi]$, let $\mathcal{I}_c = \{i \ :\ \mathrm{color}(i)=c \}$ be the set of agents with color $c$. Consider a color $c$, and rename agents such that agent $i_c$ is the agent with the $i$th highest value $v_{i_c}(\s)$ among agents in $\mathcal{I}_c$.  For every $\ell_c \in \mathcal{I}_c$ such that $\ell<i$, %\shuran{The notation is a little bit confusing here. I think it is okay to drop $c$ and just say that $i$ is the agent with $i$-th highest value in $N_c$, and maybe use $j$ when defining $N_c$ and only consider agents in $N_c$ afterward. But I can understand the current version after reading more carefully.} 
	$v_{\ell_c}(1_{i_c},\s_{-i_c})\geq v_{\ell_c}(\s)\geq v_{i_c}(\s)$. Since $\mathrm{color}(\ell_c)=\mathrm{color}(i_c)$, $\ell_c\not\rightarrow i_c$. Therefore, it must be the case that $v_{\ell_c}(0_{i_c},\s_{-i_c})\geq (1-1/\ln n) v_{i_c}(\s),$ and $\ell_c\in D_{i_c}\Rightarrow \ell_c \in  D_{i_c}\cap \mathcal{I}_c$. Hence, 
	\begin{eqnarray}
		x_{i_c} & = & \frac{1}{\chi}\frac{1}{1+\ln n}\frac{1}{|D_{i_c}|+1}\nonumber\\
		& \leq& \frac{1}{\chi}\frac{1}{1+\ln n}\frac{1}{|D_{i_c}\cap \mathcal{I}_c|+1}\nonumber\\
		& \leq&  \frac{1}{\chi}\frac{1}{1+\ln n}\frac{1}{|\{\ell_c\in \mathcal{I}_c \ : \ \ell < i\}|+1} \nonumber\\
		& = &  \frac{1}{\chi}\frac{1}{1+\ln n}\frac{1}{i}.\label{eq:alpha_bound}
	\end{eqnarray} 
	
	We have that 
	\begin{eqnarray*}
		\sum_{i=1}^n x_i  & = & \sum_{c\in [\chi]}\sum_{i=1}^{|\mathcal{I}_c|} x_{i_c} \\
		& \leq & \sum_{c\in [\chi]}\sum_{i=1}^{|\mathcal{I}_c|}\frac{1}{\chi}\frac{1}{1+\ln n}\frac{1}{i}\\
		& = & \frac{1}{\chi}\sum_{c\in[\chi]}\frac{1}{1+\ln n}\sum_{i=1}^{|\mathcal{I}_c|}\frac{1}{i}\\
		& \leq &\frac{1}{\chi}\sum_{c\in[\chi]}\frac{1+\ln|\mathcal{I}_c|}{1+\ln n} \ \leq \ 1,
	\end{eqnarray*}
	where the first inequality follows Equation~\eqref{eq:alpha_bound}. This concludes the proof.
\end{proof}

To bound $\chi(G_{\v,\s})$, we use the following graph-theoretic fact tying the maximum out-degree in a directed graph to its coloring number. We defer its proof to Appendix~\ref{sec:outdegree-chi-proof}.
\begin{lemma} [Folklore]
	A directed graph $G$ where the out-degree of each node is bounded by $k$  has a coloring number  $\chi(G)\leq 2k+1$. \label{lem:outdegree-chi}
\end{lemma}

Finally, the following lemma implies that setting $\chi=2\log n + 1$ in Equation~\eqref{eq:allocation} ensures feasibility.
\begin{lemma}
	$\chi(G_{\v,\s})\leq 2\log n + 1$. \label{lem:chi-bound}
\end{lemma}
\begin{proof}
	We show that for every vertex $j$ in $G_{\v,\s}$, $\deg^+(j)\leq \log n$. Using Lemma~\ref{lem:outdegree-chi}, this concludes the proof. Consider agent $j$, and set $P_j = \{i  \ : \ j\rightarrow i\}.$  By the definition of $G_{\v,\s}$, showing $\deg^+(j)\leq  \log n$ is equivalent to showing $|P_j|\leq \log n$. Assume toward a contradiction that there exists an agent $j$ for which $|P_j|> \log n$. Let $i\in P_j$. Since $j\rightarrow i$, $v_j(\s) \geq v_i(\s)$, and $v_j(0_i,\s_{-i}) < (1-1/\log n)v_i(\s)$. Therefore, 
	\begin{eqnarray*}
		v_j(\s)-v_j(0_i,\s_{-i}) \ >\ v_j(\s)- (1-1/\log n)v_i(\s) \ \geq\ v_j(\s)- (1-1/\log n)v_j(\s) \ =\ v_j(\s)/\log n.\label{eq:diff_lb}  
	\end{eqnarray*}
	We have that 
	\begin{eqnarray}
		\sum_{i\in P_j} v_j(\s)-v_j(0_i,\s_{-i}) \ >\ |P_j| v_j(\s)/\log n \ >\ v_j(\s), \label{eq:geq_vj}
	\end{eqnarray}
	where the second inequality uses the assumption that $|P_j| > \log n.$ Rename agents in $P_j$ to be agents $1,\ldots, |P_j|$ (and agent $j$ to be larger than $|P_j|$). Notice that by submodularity of $v_j$, for every $i\in[P_j]$,  $$v_j(\s)-v_j(0_i,\s_{-i})\leq v_j(0_{[i-1]},\s_{-[i-1]})-v_j(0_{[i]},\s_{-[i]}),$$
	and therefore
	\begin{eqnarray*}
		\sum_{i\in P_j} v_j(\s)-v_j(0_i,\s_{-i}) & = & \sum_{i = 1}^{|P_j|} v_j(\s)-v_j(0_i,\s_{-i}) \\
		& \leq & \sum_{i = 1}^{|P_j|}v_j(0_{[i-1]},\s_{-[i-1]})-v_j(0_{[i]},\s_{-[i]}) \\
		& = &  v_j(\s) - v_j(0_{P_j},\s_{-P_j}) \ \leq \ v_j(\s),
	\end{eqnarray*}
	contradicting Equation~\eqref{eq:geq_vj}. This concludes the proof.
\end{proof}

By Lemmas~\ref{lem:main-epicir}, \ref{eq:chi-fesaible} and \ref{lem:chi-bound}, we can define 
\begin{eqnarray}
	x_i(\v,\s) = \frac{1}{2\log n + 1} \cdot 
	\frac{1}{\ln n + 1} \cdot \frac{1}{|D_i(\v,\s)| + 1}.\label{eq:feasible-alpha}
\end{eqnarray}
\subsection{Approximation} \label{sec:approx}

We now show that our allocation rule gives an $O(\log^2n)$-approximation. 
%A little more formally, 
To reason about the allocation probability of the more important agents, we divide the value space into levels.
Let $\opt = {\max}_{i\in [n]}v_i(\s)$. Consider the following infinite sequence of thresholds, defined only for the sake of analysis:  $$T_0 = \opt, T_1 = \opt(1-1/\log n), \ldots,T_\ell =\opt(1-1/\log n)^\ell, \ldots$$
%, starting with the highest level $T_0 = v_{i^*}(\s)$, and with every consecutive level decreased by a multiplicative factor of $(1-1/\log n)$.  
Consider some level $\ell$ at value $T_\ell = (1-1/\log n)^\ell \cdot v_{i^*}(\s)$. For any agent $i$ with value $v_i(\s) \in [T_\ell, T_{\ell-1})$, no agent $j$ with value $v_j(\s) < (1-1/\log n)\cdot T_\ell = T_{\ell + 1}$ can be in the set $D_i$ by definition. Therefore, when reasoning about the allocation of agents with value larger than $T_\ell$, we only need to care about agents with value larger than $T_{\ell+1}$. 
We proceed by showing that one of two things might happen. 
\begin{enumerate}
	\item Either the number of agents with value between $T_{\ell+1}$ and $T_\ell$  is larger than the number of agents with value larger than $T_\ell$; or
	\item We can show that the agents with value larger than $T_\ell$ get a large fraction of the probability, and thus we can relate the achieved welfare to the value $T_\ell$.
\end{enumerate}
This is formalized by Lemma~\ref{lem:beta-bound} below. The first case implies that the number of agents doubled between $T_\ell$ and $T_{\ell+1}$. Therefore, if we start at level $T_0$ and continue down levels until we reach the second case, we are guaranteed to stop at some level $\ell \leq \log n$, as this case ensures a large fraction of probability. This allows us claim that we get a good fraction of allocation probability for agents with value larger than $T_\ell \geq (1-1/\log n)^{\log n} v_{i^*}(\s) \approx \opt/e.$

For agent $i$, let $$\ell(i) \ =\ \min\{\ell \ : \ v_i(\s)\geq T_\ell\}$$ be the index of the largest threshold that $i$ passes. We define $E_i(\v,\s)$ to be the set of agent $j$ who has the highest possible valuation $v_j(1_i,\s_{-i})$ no less than $T_{\ell(i)}$ and the lowest possible valuation $v_j(0_i, \s_{-i})$ no less than $T_{\ell(i)+1}$, assuming that $s_i$ is unknown,%We define the following set:
$$E_i(\v,\s) \ =\  \left\{j \ : \ v_j(1_i,\s_{-i}) \geq T_{\ell(i)} \ \wedge \ v_j(0_i, \s_{-i})\geq T_{\ell(i)+1}\right\}.$$
The sequence of thresholds and the set $E_i$ are illustrated in Figure~\ref{fig:ei}. We define $y_i(\v,\s)$ to be the same as $x_i(\v,\s)$, but with $E_i$ used in place of $D_i$,
%and the following quantity:
\begin{eqnarray}
	y_i(\v,\s) \ =\  \frac{1}{2\log n + 1} \cdot
	\frac{1}{\ln n + 1} \cdot \frac{1}{|E_i(\v,\s)| + 1}.\label{eq:allocationlb}
\end{eqnarray}

The $y_i$'s will be used to provide a lower bound for the probabilities of allocation, which will be easier to reason about than using $x_i$'s directly.

\begin{lemma}
	For every $\v,\s,$ and $i$, $y_i(\v,\s)\leq x_i(\v,\s).$ \label{lem:alpha-ge-beta}
\end{lemma}
\begin{proof}
	Notice that proving the lemma is equivalent to proving $|D_i(\v,\s)|\leq |E_i(\v,\s)|$. We show that $D_i(\v,\s)\subseteq E_i(\v,\s)$. Let $j$ be some agent in $D_i(\v,\s)$. This implies that $v_j(1_i,\s_{-i})\geq v_i(\s)\geq T_{\ell(i)},$ where the second inequality follows from the definition of $\ell(i)$. This also implies that $v_j(0_i,\s_{-i})\geq (1-1/\log n)v_i(\s)\geq (1-1/\log n)T_{\ell(i)}=T_{\ell(i)+1}.$ Therefore, $j\in E_i(\v,\s)$, which completes the proof.
\end{proof}

Let $N_k$ be the number of bidders with value in $[T_k,T_{k-1})$:
$$
N_k = |\{ j \ :\  \ell(j)=k\}|.
$$
The following lemma gives a bound on the quantity $y$ that each agent gets, which we use to relate to the welfare that agents achieve. The lemma looks at the number of bidders above each threshold $T_k$, $\sum_{\ell=0}^k N_\ell$. If at a certain threshold $T_{k^*}$, the number of bidders in the next layer $N_{k^*+1}$ is smaller than the number of bidders above this threshold, then the probability of an agent above $T_{k^*}$ getting allocated can be lower-bounded, as shown in the following lemma. An illustration of the $N_k$'s and $k^*$ is given in Figure~\ref{fig:ei}. 
\begin{lemma}
	Let $k^*$ be the minimal index such that 
	$$
	N_{k^*+1} \ < \ \sum_{\ell=0}^{k^*} N_{\ell}, 
	$$
	then $$\sum_{i=1}^n y_i(\v,\s)\cdot  v_i(\s)\ \geq\ \frac{(1-1/\log n)^{k^*}}{2(2\log n+1)(\ln n+1)}\opt.$$  \label{lem:beta-bound}
\end{lemma} 
\begin{proof}
	Let $k^*$ be the minimal index such that $N_{k^*+1} < \sum_{\ell=0}^{k^*} N_{\ell}$, then for any bidder $i$  with value $v_i(\s) \geq T_{k^*}$,
	\begin{eqnarray*}
		y_i(\v,\s) & = & \frac{1}{2\log n+1}\frac{1}{\ln n+1}  \frac{1}{|E_i(\v,\s)|+1}\\
		& = & \frac{1}{2\log n+1}\frac{1}{\ln n+1}  \frac{1}{|\left\{j \ : \ v_j(1_i,\s_{-i}) \geq T_{\ell(i)} \wedge v_j(0_i, \s_{-i})\geq T_{\ell(i)+1}\right\}|+1}\\
		& \geq & \frac{1}{2\log n+1}\frac{1}{\ln n+1}  \frac{1}{|\left\{j \ : \ v_j(1_i,\s_{-i}) \geq T_{k^*} \wedge v_j(0_i, \s_{-i})\geq T_{k^*+1}\right\}|+1} \\
		& \geq & \frac{1}{2\log n+1}\frac{1}{\ln n+1}  \frac{1}{|\left\{j \ : \ v_j(0_i, \s_{-i})\geq T_{k^*+1}\right\}|+1}\\
		& \geq & \frac{1}{2\log n+1}\frac{1}{\ln n+1}  \frac{1}{|\left\{j \ : \ v_j(\s)\geq T_{k^*+1}\right\}|+1} \\ 
		& = & \frac{1}{2\log n+1}\frac{1}{\ln n+1} \cdot \frac{1}{\sum_{\ell=0}^{k^*+1} N_{\ell}+1}\\
		& \geq & \frac{1}{2\log n+1}\frac{1}{\ln n+1}\cdot  \frac{1}{2 \sum_{\ell=0}^{k^*} N_{\ell}}.
	\end{eqnarray*}
	In the above, the first three inequalities follow since each set in the denominator contains the previous set. The last inequality follows the definition of $k^*$. 
	
	By definition of $T_{k^*}$, we have 
	\begin{eqnarray*}
		T_{k^*} = (1-1/\log n)^{k^*} \opt.
	\end{eqnarray*}
	We get that 
	\begin{eqnarray*}
		\sum_{i=1}^n y_i(\v,\s)\cdot  v_i(\s) & \geq & \sum_{i\ :\ v_i(\s)\geq T_{k^*}} y_i(\v,\s)\cdot v_i(\s) \\
		& \geq & \sum_{i\ :\ v_i(\s)\geq T_{k^*}}\frac{(1-1/\log n)^{k^*} \opt}{(2\log n+1)(\ln n+1)(2 \sum_{\ell=0}^{k^*} N_{\ell})}\\
		& = & \left(\sum_{\ell=0}^{k^*} N_{\ell}\right)\cdot \frac{(1-1/\log n)^{k^*} \opt}{(2\log n+1)(\ln n+1)(2 \sum_{\ell=0}^{k^*} N_{\ell})}\\
		& = & \frac{(1-1/\log n)^{k^*}}{2(2\log n+1)(\ln n+1)}\opt.
	\end{eqnarray*}
\end{proof}

We bound the value of $k^*$ from the previous Lemma.
\begin{lemma}
	Let $k^*$ be the minimal index such that $N_{k^*+1} <  \sum_{\ell=0}^{k^*} N_{\ell},$ then $k^*\leq \log n$. \label{lem:k-bound}
\end{lemma}
\begin{proof}
	For each $k$ such that $N_{k+1}  \geq \sum_{\ell=0}^{k} N_{\ell},$ the number of agents with value above $T_{k+1}$ is at least twice the number of agents with value above $T_{k}$. Since the number of agents can double at most $\log n$ times, the lemma follows.
\end{proof}

We now prove our main theorem.
\begin{theorem}
	The mechanism is $O(\log^2 n)$ approximation to the optimal welfare.
\end{theorem}
\begin{proof}
	The mechanism achieves expected welfare	of
	\begin{eqnarray*}
		\sum_{i=1}^n x_i(\v,\s)\cdot v_i(\s) & \geq & \sum_{i=1}^n y_i(\v,\s)\cdot v_i(\s) \\
		& \geq & \frac{(1-1/\log n)^{k^*}}{2(2\log n+1)(\ln n+1)}\opt \\
		& \geq & \frac{(1-1/\log n)^{\log n}}{2(2\log n+1)(\ln n+1)}\opt \\
		& \approx & \frac{\opt}{2e(2\log n+1)(\ln n+1)},
	\end{eqnarray*}
	where the first inequality follows Lemma~\ref{lem:alpha-ge-beta}, the second inequality follows Lemma~\ref{lem:beta-bound}, and the third inequality follows Lemma~\ref{lem:k-bound}.
\end{proof}

\paragraph{Remark about implementation using queries} In presenting the mechanism, we ask that each bidder bids their valuation to the mechanism. While sending the representation of an arbitrary function can be demanding, we note that the mechanism can actually be implemented using $2n-1$ queries to each buyer's valuation function. For every buyer $i$, we only need to access $v_i(\s)$, and $\{v_i(0_j, \s_{-j}), v_i(1_j, \s_{-j})\}$ for every $j\neq i$. This enough to determine the allocation \textit{and} payments, since an agent $i$ needs to know only $\{v_j(0_i,\s_{-i}),v_j(1_i,\s_{-i})\}_{j\neq i}$ in order to know the effect of increasing their value on the set $D_i(\v,\s)$, which in turn determines $i$'s allocation probability, giving enough information to compute the payments using Equation~\eqref{eq:payment}.

\paragraph{Extension to multi-dimensional signals} While the mechanism and analysis is presented in the context of single-dimensional signals, we can extend it to the domain of multi-dimensional signals. It would require that when defining the mechanism, we define the set $D_i$ as the set of bidders that satisfy the same condition with all of $i$'s signals getting their highest/lowest possible value. The proof of Lemmas~\ref{lem:chi-bound} will reason about zeroing out a set of signals of an agent, but submodularity carries on for sets of signals as well, and the proof will mostly be identical. Finally, for multi-dimensional signals, only the sufficient direction of the truthfulness characterization (Proposition~\ref{prop:epicir_suf}) holds, but this is what we need to show the mechanism is EPIC-IR.

\paragraph{Extension to $d$-SOS condition from Eden et al.~\cite{EFFGK}} Eden et al.~\cite{EFFGK} extend the $SOS$ condition to settings where the inequality in Equation~\eqref{def:sos} holds up to a $d$ multiplicative factor. We note that under $d$-SOS, we can adapt the mechanism to get an $O(d\log^2 n)$-approximation to the optimal welfare. The only change in the mechanism is by setting $\chi$ in Equation~\eqref{eq:allocation} to be $2d\log n+1$. This is since once can show that under  $d$-SOS, $\chi(G_{\v,\s})\leq 2d\log n +1$. This follows by changing the proof of Lemma~\ref{lem:chi-bound} to consider a vertex $j$ of out-degree larger than $d\log n$, and showing that under $d$-SOS, the left-hand side of Equation~\eqref{eq:geq_vj} is strictly larger than $d\cdot v_j(\s)$, and smaller than $d\cdot v_j(\s)$, deriving a contradiction.
\section{$\Theta(k)$-approximation for $k$-bounded Dependency}\label{sec:k-depend}

Consider the case where for each bidder $i$ there is a set $$B_i = \{j\neq i \ : \ \exists \s \mbox{ such that } v_i(\s) > v_i(0_j,\s_{-j})\},$$
that is,  $i$ depends only on the signals of bidders in $B_i$. 
%$B_i$ of at most $k$ bidders such that for every $j\notin B_i$ and $\s$, $v_i(\s)=v_i(0,\s_{-j})$, that is, $i$ depends only on the signals of bidders in $B_i$. \shuran{do we want $B_i$ to include $i$?} 
Consider the following dependency graph $G=(V,E)$ with $n$ nodes, one for each bidder $i$: there is an edge from $i$ to $j$ if $j\in B_i$. Let $\deg^+(i)$ be the out-degree of node $i$ in this graph, and let $k=\max_i \deg^+(i) = \max_i|B_i|$, which means that each bidder's value depends on at most $k$ other bidders' signals.  Consider the following mechanism: 
\begin{enumerate}
	\item Ask the bidders to report signals and valuations, $\tilde{s}_i$ and $\tilde{v}_i$.
	\item For bidder $i$, use all other bidders' reports except $i$'s to estimate the other bidders' valuations %\shuran{do we want to use $t$ or $T$?}
	$$
	T^{(i)}_j = \tilde{v}_j(\tilde{s}_{-i}, 0), \quad \forall j\neq i. 
	$$
	%   Use other bidders' reported signals and bidder $i$'s reported valuation to estimate bidder $i$'s valuation
	%   $$
	%   t^{(i)}_i = \tilde{v}_i(\tilde{s}_{-i}, 0).
	%   $$
	Bidder $i$ is one of the candidates if his valuation based on the reported messages is higher than these estimates for $i$ of all other bidders' estimated valuations
	$$
	\tilde{v}_i(\tilde{s}) \ge T^{(i)}_j,\quad \forall j\neq i,
	$$
	breaking ties by lexicographical order.% where if $\tilde{v}_i(\tilde{s}) = t^{(i)}_j$, we say that  $\tilde{v}_i(\tilde{s}) \ge t^{(i)}_j$ if $i<j$ (that is, break ties according to index).
	\item For each candidate $i$, allocate to $i$ with probability 
	\begin{align*}
		\frac{1}{2}\cdot \frac{1}{\max_{j\ :\ j\rightarrow i} \deg^+(j)+1}. %\quad \text{\sz{if there exists $j$ such that $(j,i)\in E$,}}		
	\end{align*}
%	\sz{and allocate to $i$ with probability $1/2$ otherwise.}
 %That is, 
 The probability is inversely proportional to the maximum out-degree of a node corresponding to an agent that depends on $i$'s signal. Charge payments according to Equation~\eqref{eq:payment}. %If bidder $i$ gets the item, his payment is his highest threshold $\max_{j\neq i} t^{(i)}_j$.	
\end{enumerate}

\begin{theorem}
	The mechanism is EPIC-IR, feasible, and gives an $2(k+1)$-approximation to the maximum social welfare. \label{thm:k-dep}
\end{theorem}
\begin{proof}
	Notice that $T^{(i)}_j$ is independent of $i$'s bid, and that $i$'s allocation is non-decreasing in $v_i(\s)$; therefore, the mechanism satisfies the conditions for Proposition~\ref{prop:epicir_suf}, which implies the mechanism can be implemented in an EPIC-IR manner.

	For feasibility, we notice the following: let $i^* = \min\{i\ : \ i\in \arg\max_{j\in [n]}v_j(\s)\}$ (the lowest index highest valued agent). We claim that only $i^*$, and agents $j$ such that there is an edge from $i^*$ to $j$ can be candidates by the mechanism analyzed in this section. The reason for this is that by the way we break ties, $i^*$ is a candidate, but for every agent $j$ such that there is no edge between $i^*$ and $j$, $$v_j(\s) \leq v_{i^*}(\s)=v_{i^*}(\s_{-i},0_i) = T_{i^*}^{(j)},$$ where the first equality follows the definition of $G$. If $v_j(\s) = T_{i^*}^{(j)}$, then this is treated as an inequality since $i^*$ is the minimal index agent with a maximal value. Therefore, $j$ cannot be a candidate. 
	For agents $j\neq i^*\in B_{i^*}$ who have an edge from $i^*$ and thus can possibly be a candidate, the probability of $j$ getting allocated is bounded by $1/2(\deg^+(i^*)+1)$. Therefore,
	\begin{eqnarray*}
		\sum_i x_i(\v,\s) &= &x_{i^*}(\v,\s) + \sum_{i\ :\ (i^*,i)\in E} x_{i}(\v,\s) \\
		&\leq &\frac{1}{2} + \sum_{i\ :\ (i^*,i)\in E}\frac{1}{2}\cdot \frac{1}{\max_{j:(j,i)\in E} \deg^+(j)+1} \\
		& \leq & \frac{1}{2} + \sum_{i\ :\ (i^*,i)\in E}\frac{1}{2}\cdot \frac{1}{ \deg^+(i^*)+1} \ < \ 1.
	\end{eqnarray*} 

	As for approximation, notice that $i^*$ is a candidate, and is allocated with probability $$\frac{1}{2}\cdot \frac{1}{\max_{j:(j,i^*)\in E} \deg^+(j)+1} > \frac{1}{2(k+1)}.$$
\end{proof}

As in the case of the previous mechanism we presented, this mechanism can be implement using a polynomial number, $n(n-1)$, of queries to the valuation oracle of the agents.
We show that our mechanism is tight up to a factor of 2.

\begin{proposition}
	For any $k$, we have the following inapproximability results:
	\begin{enumerate}
		\item With public valuations where no EPIC-IR mechanism can achieve a better approximation than $k+1$.
		\item With private valuations that satisfy the single-crossing condition where no EPIC-IR mechanism can achieve a better approximation than $k+1$.
	\end{enumerate}
\end{proposition}
\begin{proof}
	To prove the first claim, consider the following instance, which is a simple adaptation of the instance in Section 3 of \cite{EFFG}. In the slightly modified instance, $$v_i(\s) = \prod_{j \in [k+1]\setminus\{i\}}s_j + \epsilon s_i$$ for $i\in [k+1]$, and $v_i(s)=\epsilon s_i$ for the rest of the agents ($\epsilon$ is arbitrarily small). It is easy to see that the $\max_i \deg^+(i)=k$. For $i\in [k+1]$. Let $\s^{i}$ be the signal profile where $s_i^i=0$, and $s_i^j=1$ for $j\neq i$. Notice that under signals $\s^i$, $i$ has value $1$, and all other agents have value an arbitrarily small value $\epsilon$. In order to get a better than $k+1$ approximation, agent $i$ has to have an allocation probability $x_i(\s^i) > 1/(k+1)$. Otherwise, decreasing $\epsilon$ could get us arbitrarily close to or smaller than $k+1$-approximation. By Proposition~\ref{prop:epicir_nec}, by monotonicity, for each $i\in [k+1]$, $x_i(\textbf{1})>1/ (k+1)$ as well, which implies $\sum_{i\in [k+1]}x_i(\s)>1$, contradicting feasibility.\footnote{The same bound holds even for the Bayesian setting, where the seller knows the distribution over signals of the agents, as proved in~\cite{EFFG}.}
	
	To prove the second claim, consider the adaptation of the instances in Proposition~\ref{prop:sc-n}:
	\begin{itemize}
			\item \textbf{Case 1.} Each agent $i\in[k+1]$ has valuation $v_i(\s)=\prod_{j=1}^{k+1} s_j+\epsilon i$, and $v_i(\s)= \epsilon s_i$ for $i\notin [k+1]$. Signal profile is $\s=\mathbf{1}$. The valuations satisfy single crossing. There must be an agent $i\in [k+1]$ who's allocation probability is at most $1/(k+1)$. Assume without loss it is agent $1$.
			
			\item \textbf{Case 2.} $v'_1=0.5+s_1$, and $v_2,\ldots, v_n$ are as before. Note that single-crossing is still satisfied. The signals are $s'_1=0$, $\s_{-1}=\textbf{1}$.  
	\end{itemize}
	It is easy to see that $\max_i \deg^+(i)=k$. By Proposition~\ref{prop:epicir_nec}, agent~$1$'s allocation probability in Case 2 is at most $1/(k+1)$. Since he's the only agent with non-negligible value, we get that no EPIC-IR mechanism can get better than $k+1$-approximation.
\end{proof}

\section{Conclusion} \label{sec:conclusion}
In this paper, we open up a new direction for studying interdependent value settings. Namely, we consider what we believe might be a more realistic model where agents' valuation functions are not known to the mechanism designer, and study EPIC-IR mechanisms in this context. We provide the first competitive bound for such a setting by proving that under the submodularity over signals condition, there exists a randomized mechanism with a welfare within $O(\log^2 n)$ of the optimal welfare. We also show that in stark contrast to the public valuations setting, only assuming single-crossing can result in terrible approximation bounds. One major open question left behind by our work is the following.

\vspace{0.3cm}
\noindent\textbf{Main open question.} Is there a constant factor approximation EPIC-IR mechanism for interdependent values with private valuation functions when the valuations satisfy the submodularity over signals condition?
\vspace{0.3cm}

We note that the only lower bound known for SOS valuations is for the public valuations setting, where Eden et al.~\cite{EFFGK} show no EPIC-IR mechanism can have a better approximation than 2, which it is not even tight for the public valuations setting.

It might also be interesting to consider other restrictions over the valuations which might lead to different insightful bounds. Considering restricted valuations can be justified since in the interdependent settings, as even in the public valuations case, unrestricted valuations give terrible guarantees~\cite{EFFG}. 
While our mechanism is prior-free, considering Bayesian settings, where the seller has distribution over possible valuation functions of the buyer might lead to improved bounds, as well as considering other objectives, such as second-best or max-min robust guarantees. Finally, one might consider more general combinatorial auction settings or a more difficult benchmark such as revenue maximization.

\paragraph{Acknowledgments} We deeply thank Yiling Chen for the helpful discussions and feedback.

\newpage

\bibliographystyle{acm}
\bibliography{masterbib.bib}

\newpage

\appendix

\section{Missing Proofs} \label{sec:prelims-proofs}

\subsection{Proof of Proposition~\ref{prop:epicir_suf}} \label{sec:epicir-proof}
	First we show the sufficient direction.  Consider the payment function 
	\begin{eqnarray*}
		p_i(\v,\s) = x_i(\v_{-i},\s_{-i},v_i(\s))v_i(\s) - \int_{0}^{v_i(\s)}x_i(\v_{-i},\s_{-i},t) dt. %\label{eq:payment}
	\end{eqnarray*}
	First, notice that for truthful $\s_{-i}$, and any $\v_{-i}$, we have that $$u_i(\v_{-i},\s_{-i};(v_i,s_i);(v_i,s_i)) = x_i(\v_{-i},\s_{-i},v_i(\s))v_i(\s) - p_i(\v.\s) = \int_{0}^{v_i(\s)}x_i(\v_{-i},\s_{-i},t) dt  \geq 0,$$
	which implies EPIR.
	
	To show EPIC, consider truthful $\s_{-i}$ (and any $\v_{-i}$), and consider arbitrary $v'_i$, $s'_i$.  EPIC implies that %If $v'_i(s'_i,\s_{-i}) < v_i(\s)$, then we have 
	\begin{eqnarray*}
		u_i(\v_{-i},\s_{-i};(v'_i,s'_i);(v_i,s_i)) & =  & x_i(\v_{-i},\s_{-i},v'_i(s'_i,\s_{-i}))v_i(\s) - p_i(\v_{-i},\s_{-i},v'_i,s'_i)\\
		& = & x_i(\v_{-i},\s_{-i},v'_i(s'_i,\s_{-i}))v_i(\s) \\& & - ( x_i(\v_{-i},\s_{-i},v'_i(s'_i,\s_{-i}))v'_i(s'_i,\s_{-i}) -  \int_{0}^{v'_i(s'_i,\s_{-i})}x_i(\v_{-i},\s_{-i},t) dt)\\
		&\leq & \int_{0}^{v_i(\s)}x_i(\v_{-i},\s_{-i},t) dt \ =\ u_i(\v_{-i},\s_{-i};(v_i,s_i);(v_i,s_i)),	
	\end{eqnarray*}
	which is equivalent to 
	$$x_i(\v_{-i},\s_{-i},v'_i(s'_i,\s_{-i}))(v_i(\s) - v'_i(s'_i,\s_{-i})) \leq    \int_{v'_i(s'_i,\s_{-i})}^{v_i(\s)}x_i(\v_{-i},\s_{-i},t) dt.$$ 
	If $v'_i(s'_i,\s_{-i}) \leq v_i(\s)$, then by monotonicity of $x_i$, the rhs is no smaller than $\int_{v'_i(s'_i,\s_{-i})}^{v_i(\s)}x_i(\v_{-i},\s_{-i},v'_i(s'_i,\s_{-i})) dt,$ which is equal to the lhs. If $v'_i(s'_i,\s_{-i}) \geq v_i(\s)$, then EPIC is equivalent to showing $$x_i(\v_{-i},\s_{-i},v'_i(s'_i,\s_{-i}))(v'_i(s'_i,\s_{-i})-v_i(\s)) \geq    \int_{v_i(\s)}^{v'_i(s'_i,\s_{-i})}x_i(\v_{-i},\s_{-i},t) dt,$$
	which again follows since, by monotonicity, the rhs is smaller than $\int_{v_i(\s)}^{v'_i(s'_i,\s_{-i})}x_i(\v_{-i},\s_{-i},v'_i(s'_i,\s_{-i})) dt$ which equals the lhs.

%Old - Alon wrote
\begin{comment}
	Let $v_i,s_i$ and $v'_i,s'_i$ be such that $v_i(s_i,\s_{-i}) > v'_i(s'_i,\s_{-i})$. Every IC mechanism must satisfy the following two inequalities:
	\begin{eqnarray*}
		v_i(s_i,\s_{-i})x_i(\v_{-i},\s_{-i};v_i,s_i) - p_i(\v_{-i},\s_{-i};v_i,s_i) \ & \geq & \  v_i(s_i,\s_{-i})x_i(\v_{-i},\s_{-i};v'_i,s'_i) - p_i(\v_{-i},\s_{-i};v'_i,s'_i)\\
		v'_i(s'_i,\s_{-i})x_i(\v_{-i},\s_{-i};v'_i,s'_i) - p_i(\v_{-i},\s_{-i};v'_i,s'_i) \ & \geq & \  v'_i(s'_i,\s_{-i})x_i(\v_{-i},\s_{-i};v_i,s_i) - p_i(\v_{-i},\s_{-i};v_i,s_i).
	\end{eqnarray*}
	
	Summing up the two inequalities and rearranging yields:
	$$(v_i(s_i,\s_{-i})-v'_i(s'_i,\s_{-i}))(x_i(\v_{-i},\s_{-i};v_i,s_i)-x_i(\v_{-i},\s_{-i};v'_i,s'_i)) \ \geq \ 0,$$
	and since $v_i(s_i,\s_{-i}) > v'_i(s'_i,\s_{-i})$, this implies $x_i(\v_{-i},\s_{-i};v_i,s_i)\geq x_i(\v_{-i},\s_{-i};v'_i,s'_i).$
\end{comment}

\subsection{Proof of Proposition~\ref{prop:epicir_nec}}

For the necessary conditions, we need some more tools. Roughgarden Talgam-Cohen~\cite{RTC} introduced an equivalent of Myerson's theorem~\cite{Myerson} for ex-post IC and IR implementation in that IDV model when the valuation functions are public. We use it to prove our characterization in the when the valuations are private.

\begin{proposition}[\cite{RTC}] \label{prop:RTC}
	In the IDV setting with public valuations, a mechanism is EPIC-IR if and only if for every $i, \s_{-i}$, $x_i(\cdot, \s_{-i})$ is monotonically non-decreasing in $s_i$, and the following payment-identity holds:
	\begin{eqnarray}
		& & p_i(\s)\ =\ x_i(\s)v_i(\s)-\int_{v_i(0,\s_{-i})}^{v_i(s_i,\s_{-i})}x_i(v_i^{-1}(t|\s_{-i}),\s_{-i})dt-(x_i(0,\s_{-i})v_i(0,\s_{-i}) - p_i(0,\s_{-i})),\nonumber\\
		& & \mbox{where } p_i(0,\s_{-i}) \ \leq\  x_i(0,\s_{-i})v_i(0,\s_{-i}).\label{eq:payment-rtc}
	\end{eqnarray}
\end{proposition}

We use the equivalent form of Equation~\eqref{eq:payment-rtc} which is more convenient for our purposes:

\begin{eqnarray}
	& & p_i(\s)\ =\ x_i(\s)v_i(\s)-\int_{v_i(0,\s_{-i})}^{v_i(s_i,\s_{-i})}x_i(v_i^{-1}(t|\s_{-i}),\s_{-i})dt - q_i(0,\s_{-i}),\nonumber\\
	& & \mbox{where } q_i(0,\s_{-i}) \ \geq\  0.\label{eq:payment-rtc2}
\end{eqnarray}

For the proof, we also use the following property.

\begin{theorem}[Lebesgue Differentiation]
	Let $f$ be an integrable function on $\mathbb{R^+}$. Then the following equality holds almost everywhere on $\mathbb{R^+}$.
	$$
	f(y) = \lim_{r\to 0^+} \frac{1}{r} \int_{[y,y+r]} f \,d\mu,
	$$
	where $\mu$ is the Lebesgue measure on $\mathbb{R^+}$.
\end{theorem}

    Now, we show the necessary direction, starting with showing that any allocation rule $x_i(\v, \s)$ must be monotone in $v_i(\s)$.

    By Definition~\ref{def:EPIC-IR} of ex-post IC, we have the following, for any signal profile $\s$, valuation functions $\v$, bidder $i$, and reports $s'_i, v'_i$,
	$$
	v_i(\s) x_i(\v, \s) - p_i(\v, \s) \ge v_i(\s) x_i(\v_{-i}, v_i', \s_{-i}, s_i') - p_i(\v_{-i}, v_i', \s_{-i}, s_i').
	$$
	%Similarly, individual rationality gives:  $$v_i(\s) x_i(\s, \v) - p_i(\s, \v) \ge 0.	$$
	
	Then for any EPIC mechanism $M=(x,p)$, we must have 
	%$$v_i(\s_{-i}, s_i) x_i(\s_{-i}, s_i, \v_{-i}, v_i) - p_i(\s_{-i}, s_i, \v_{-i}, v_i) \ge v_i(\s_{-i}, s_i) x_i(\s_{-i}, \hat s_i, \v_{-i}, \hat v_i) - p_i(\s_{-i}, \hat s_i, \v_{-i}, \hat v_i).$$
	$$
	v_i(\s) \bigg  ( x_i(\v, \s) - x_i(\v_{-i}, \hat v_i, \s_{-i}, \hat s_i)  \bigg ) \ge p_i(\v, \s)  - p_i(\v_{-i}, \hat v_i, \s_{-i}, \hat s_i).
	$$
	and
	$$
	p_i(\v, \s)  - p_i(\v_{-i}, \hat v_i, \s_{-i}, \hat s_i) \ge \hat v_i(\s_{-i}, \hat s_i) \bigg  ( x_i(\v, \s) - x_i(\v_{-i}, \hat v_i, \s_{-i}, \hat s_i)  \bigg ).
	$$
	That is,
	$$
	v_i(\s) \bigg  ( x_i(\v, \s) - x_i(\v_{-i}, \hat v_i, \s_{-i}, \hat s_i)  \bigg ) \ge \hat v_i(\s_{-i}, \hat s_i) \bigg  ( x_i(\v, \s) - x_i(\v_{-i}, \hat v_i, \s_{-i}, \hat s_i)  \bigg ).
	$$	
	Then in order to satisfy this inequality, whenever it is the case that $v_i(\s) > \hat v_i(\s_{-i}, \hat s_i)$, it must also be the case that $x_i(\v, \s) \geq x_i(\v_{-i}, \hat v_i, \s_{-i}, \hat s_i)$, which proves the first condition. %Therefore $x_i(\v, \s)$ must be monotone non-decreasing in $v_i(s)$. 

\begin{comment}
    Then if $v_i(\s) = \hat v_i(\s_{-i}, \hat s_i)$, it must be that $x_i(\v, \s) = x_i(\v_{-i}, \hat v_i, \s_{-i}, \hat s_i)$.
    \shuran{The inequality is equivalent to 
    $$
    \bigg(v_i(\s)- \hat v_i(\s_{-i}, \hat s_i)\bigg) \bigg  ( x_i(\v, \s) - x_i(\v_{-i}, \hat v_i, \s_{-i}, \hat s_i)  \bigg ) \ge 0.
    $$
    If $v_i(\s) = \hat v_i(\s_{-i}, \hat s_i)$, then the inequality trivially holds?}
    This concludes that (i) and (ii) are necessary and thus the proof of Proposition~\ref{prop:epicir}.
\end{comment}

    We then prove the second condition. Consider any $i$ and fixed $\v_{-i}, \s_{-i}$. By our assumption, $v_i(s_i, \s_{-i})$ is a monotonically increasing function of $s_i$. Define $v_i^{-1}(\cdot)$ to be the inverse function of $v_i(\cdot, \s_{-i})$ for the fixed $\s_{-i}$ and define function $x_{v_i}: \mathbb{R}^+ \to \mathbb{R}^+$ with $$x_{v_i}(y) = x_i(\v_{-i}, v_i, s_{-i}, v_i^{-1}(y)),$$ which maps a possible valuation $y$ of bidder $i$ to the corresponding allocation $x_i$ when bidder $i$'s valuation function is $v_i$. We assume that $x_{v_i}$ is an integrable function on $[v_i(0, \s_{-i}),+\infty)$. Then we show that for any $v_i, \hat v_i \in V_i$, defining $L = \max\{v_i(0, \s_{-i}),\hat v_i(0,\s_{-i})\}$, we have
    \begin{align*}
    x_{v_i}(\cdot) = x_{\hat v_i} (\cdot) \text{ almost everywhere on } [L, +\infty), 	
    \end{align*}
    that is, the set of points $y\in [L, +\infty)$ at which $x_{v_i}(y) \neq x_{\hat v_i} (y)$ has  measure zero. % Formally, for any $\varepsilon>0$, there exists a sequence of intervals $I_1, I_2, \dots$ in $\mathbb{R}^+$ with total length $\le \varepsilon$ and 
   % $$
   % \{y: f_{v_i}(y) \neq f_{\hat v_i} (y)\} \subseteq I_1 \cup I_2 \cup \cdots. 
   % $$

	We prove the statement by showing that 
	\begin{align*} 
	\int_{L}^y x_{v_i}(t) \, dt - \int_{L}^y x_{\hat v_i}(t) \, dt = \text{a constant, for all } y>L.
	\end{align*}
	Note that for fixed $\s_{-i}, \v_{-i}$ and $v_i(\cdot)$, the allocation function $x_{v_i}(\cdot)$ becomes a single-argument function. We can thus use the well-known payment identity function given by~\cite{Myerson} and~\cite{RTC} (Proposition~\ref{prop:RTC}). According to Proposition~\ref{prop:RTC}, the payment rule of a ex-post IC mechanism  $(x,p)$ must satisfy the following for fixed $\v_{-i}, \s_{-i}$ and  fixed valuation functions $v_i(\cdot), \hat v_i(\cdot)$,
	\begin{align}
	p_i(\v_{-i},v_i, \s_{-i},v_i^{-1}(y))\ = y\cdot x_{v_i}(y) -\int_{v_i(0,\s_{-i})}^{y}x_{v_i}(t)\, dt -x_i(\v_{-i}, v_i, \s_{-i}, 0)v_i(0,\s_{-i}) + p_i(\v_{-i}, v_i, \s_{-i},0)\label{eqn:identity1}\\
	%p_i(\v_{-i},v_i, \s_{-i},v_i^{-1}(y)) = p_i(\v_{-i}, v_i, \s_{-i},v_i^{-1}(0)) + y\cdot f_{v_i}(y) - \int_0^y f_{v_i}(t)\, dt  \\
	p_i(\v_{-i}, \hat v_i, \s_{-i},\hat v_i^{-1}(y)) =   y\cdot x_{\hat v_i}(y) - \int_{\hat v_i(0, \s_{-i})}^y x_{\hat v_i}(t)\, dt  -x_i(\v_{-i}, \hat v_i, \s_{-i}, 0)\hat v_i(0,\s_{-i}) + p_i(\v_{-i}, \hat v_i, \s_{-i}, 0) \label{eqn:identity2}
	\end{align}
	for all $y>L$.
	In addition, consider a bidder $i$ who has valuation function $v_i(\cdot)$ and private signal $s_i = v_i^{-1}(y)$ so that her value equals $y$. Due to ex-post IC, bidder $i$'s utility should not increase when she misreport $\hat v_i(\cdot)$ and signal $s_i = \hat v_i^{-1}(y)$, i.e.,
	\begin{align*}
	y\cdot x_{v_i}(y)-p_i(\v_{-i},v_i, \s_{-i},v_i^{-1}(y)) \ge y\cdot x_{\hat v_i}(y) - p_i(\v_{-i}, \hat v_i, \s_{-i},\hat v_i^{-1}(y)) 
	\end{align*}
	The same should hold for a bidder with valuation function $\hat v_i(\cdot)$ and signal $s_i = \hat v_i^{-1}(y)$. Her utility should not increase if she misreport $v_i(\cdot)$ and signal $s_i = v_i^{-1}(y)$, i.e., 
	\begin{align*}
	y\cdot x_{v_i}(y)-p_i(\v_{-i},v_i, \s_{-i},v_i^{-1}(y)) \le y\cdot x_{\hat v_i}(y) - p_i(\v_{-i}, \hat v_i, \s_{-i},\hat v_i^{-1}(y)).
	\end{align*}
	Therefore, bidder $i$'s utility for these two cases should be the same,
	\begin{align}
	y\cdot x_{v_i}(y)-p_i(\v_{-i},v_i, \s_{-i},v_i^{-1}(y)) = y\cdot x_{\hat v_i}(y) - p_i(\v_{-i}, \hat v_i, \s_{-i},\hat v_i^{-1}(y)) \label{eqn:identity3}.
	\end{align}
	because otherwise bidder $i$ can misreport in one of the two cases to improve the expected utility. Then if we subtract \eqref{eqn:identity1} with \eqref{eqn:identity2} and combine it with \eqref{eqn:identity3}, we get
	\begin{align*}
		&\int_{v_i(0, \s_{-i})}^y x_{v_i}(t) \, dt - \int_{\hat v_i(0, \s_{-i})}^y x_{\hat v_i}(t) \, dt \\
		 = &-x_i(\v_{-i}, v_i, \s_{-i}, 0)v_i(0,\s_{-i}) + p_i(\v_{-i}, v_i, \s_{-i},0) + x_i(\v_{-i}, \hat v_i, \s_{-i}, 0)\hat v_i(0,\s_{-i}) - p_i(\v_{-i}, \hat v_i, \s_{-i}, 0)\\
		= & \text{ a constant}.
	\end{align*}
	This holds for all $y>L$. This difference should still be a constant if we truncate the integrals at point $L$, that is, 
	\begin{align}
	& \int_{L}^y x_{v_i}(t) \, dt - \int_{L}^y x_{\hat v_i}(t) \, dt \notag \\
	=&	\left(\int_{v_i(0, \s_{-i})}^y x_{v_i}(t) \, dt - \int_{v_i(0, \s_{-i})}^L x_{v_i}(t) \, dt\right) - \left(\int_{\hat v_i(0, \s_{-i})}^y x_{\hat v_i}(t) \, dt -\int_{\hat v_i(0, \s_{-i})}^L x_{\hat v_i}(t) \, dt \right) \notag \\
	=& \left(\int_{v_i(0, \s_{-i})}^y x_{v_i}(t) \, dt - \int_{\hat v_i(0, \s_{-i})}^y x_{\hat v_i}(t) \, dt\right) -\int_{v_i(0, \s_{-i})}^L x_{v_i}(t)  +\int_{\hat v_i(0, \s_{-i})}^L x_{\hat v_i}(t) \, dt \notag\\
	= & \text{ a constant}. \label{eqn:equal}
	\end{align}

	We then apply the Lebesgue Differentiation Theorem to show that $x_{v_i}(\cdot) = x_{\hat v_i} (\cdot)$ almost everywhere on $[L, +\infty)$. Let $g(y) = x_{v_i}(y) - x_{\hat v_i}(y)$ for $y\in [L, +\infty)$. According to \eqref{eqn:equal}, we have $\int_{L}^y g(t)\, dt = \text{a constant}$ for all $y>L$, which means that $\int_{y}^{y+r} g(t)\, dt = 0$ for all $y>L$ and $r\ge 0$.  Then according to the Lebesgue Differentiation Theorem,
   \begin{align*}
   g(y) = 	\lim_{r\to 0^+} \frac{1}{r} \int_{[y,y+r]} g \,d\mu = 0 
   \end{align*}
	almost everywhere on $[L, +\infty)$, where $\mu$ is the Lebesgue measure on $[L, +\infty)$. This completes our proof.

\subsection{Proof of Lemma~\ref{lem:outdegree-chi}} \label{sec:outdegree-chi-proof}
We prove by induction on the number of vertices $n$ in the graph. For $n\leq 2k+1$, the graph is trivially $2k+1$ colorable. Assume the claim is true for $n-1$ vertices, and consider a graph of $n$ vertices. The number of edges in such a graph is at most $nk$ (each vertex adds at most $k$ new edges), and therefore, the sum of %(in+out) 
degrees (including both in-degrees and out-degrees) is at most $2nk$. By an averaging argument, there exists a vertex with degree at most $2k$. Consider this vertex $v$, and the graph $G'=G\setminus\{v\}$. By the induction hypothesis, $G'$ is $2k+1$-colorable, but $v$ has at most $2k$ neighbors. Therefore, there is a color that we can assign to $v$ which is different than all its neighbors.

\end{document}